\documentclass[11pt]{article}
\usepackage{graphicx}                  

\usepackage{graphicx}
\usepackage{epstopdf}
\usepackage[multiple]{footmisc}
\usepackage{eurosym}
\usepackage{authblk}
 
\usepackage{hyperref}   

\usepackage{amsmath,amssymb,amsthm,mathrsfs,amsfonts,dsfont} 

\usepackage{amsfonts}
\usepackage{array}
\usepackage{amsthm}
\usepackage{amsmath}
\usepackage{setspace}
\usepackage{algorithm,algorithmic}
\usepackage{natbib}

\newcommand*\colvec[3][]{
    \begin{pmatrix}\ifx\relax#1\relax\else#1\\\fi#2\\#3\end{pmatrix}
}

\newcommand*\colvecfive[5][]{
    \begin{pmatrix}\ifx\relax#1\relax\else#1\\\fi#2\\#3 \\#4\\#5\end{pmatrix}
}

\usepackage[caption = false]{subfig}
\usepackage[font=small,skip=0pt]{caption}

\theoremstyle{remark}

\newtheorem{definition}{Definition}
\theoremstyle{definition}

\newtheorem{proposition}{Proposition}
\theoremstyle{proposition}

\theoremstyle{lemma}

\theoremstyle{theorem}

\DeclareMathOperator{\E}{\mathbb{E}}

\usepackage{titlesec}

\titleformat{\paragraph}
{\normalfont\normalsize\bfseries}{\theparagraph}{1em}{}
\titlespacing*{\paragraph}
{0pt}{3.25ex plus 1ex minus .2ex}{1.5ex plus .2ex}

\usepackage{rotating}
\usepackage{tikz}

\begin{document}

\title{Collateral Unchained: Rehypothecation networks, concentration and systemic effects}

\author[1,2]{Duc Thi Luu}
\author[2,3]{Mauro Napoletano \thanks{Corresponding author: \texttt{mauro.napoletano@sciencespo.fr}}}
\author[4,5]{Paolo Barucca}
\author[5]{Stefano Battiston}
\affil[1]{Department of Economics,  University of Kiel,  Germany}
\affil[2]{OFCE-Sciences Po}
\affil[3]{SKEMA business school and Universit\'{e} C\^{o}te d’Azur}
\affil[4]{London Institute for Mathematical Sciences, London, UK}
\affil[5]{ Department of Banking and Finance, University of Zurich, Switzerland}

\date{January 31, 2018}

\maketitle

\begin{abstract} 
\noindent We study how network structure affects the dynamics of collateral in presence of rehypothecation. 
We build a simple model wherein banks interact via chains of repo
contracts and use their proprietary collateral or re-use the
collateral obtained by other banks via reverse repos. In this
framework, we show that total collateral volume and its velocity
are affected by characteristics of the network like the length of
rehypothecation chains, the presence or not of chains having a cyclic
structure, the direction of collateral flows, the density of the
network. In addition, we show that structures where collateral flows
are concentrated among few nodes (like in core-periphery networks) allow 
large increases in collateral volumes already with small network density. 
Furthermore, we introduce in the model collateral hoarding rates determined 
according to a Value-at-Risk (VaR) criterion, and we then study the
emergence of collateral hoarding cascades in different networks. Our
results highlight that network structures with highly concentrated
collateral flows are also more exposed to large collateral hoarding
cascades following local shocks. These networks are therefore characterized
by a trade-off between liquidity and systemic risk. 
\medskip



\end{abstract}

\newpage

\onehalfspacing

\section{Introduction}
\label{sec:intro}

This paper investigates the collateral dynamics when banks are connected in a
network of financial contracts and they have the ability to rehypothecate the collateral along chains of contracts.
Collateral is of increasing importance for the functioning of the
global financial system. One reason for this is that the non-bank/bank
nexus has become considerably more complex over the past two decades,
in part because the separation between hedge funds, mutual funds,
insurance companies, banks, and broker/dealers has become blurred as a
result of financial innovation and deregulation
\citep[][]{singh2016collateral,pozsar2011nonbank}. Another reason for
the significant increase in collateral volumes \citep[comparable to M2
until the recent financial crisis, see e.g.][]{singh2011velocity} has
been the diffusion of rehypothecation agreements. The role of
collateral in lending agreements is to protect the lender against a
borrower's default. Rehypothecation\footnote{Throughout this paper,
  the terms ``re-use'', ``rehypothecation'', and ``re-pledge'' are
  interchangeable.} consists in the right of the lender to re-use the
collateral to secure another transaction in the future
\citep[see][]{monnet2011rehypothecation}.

Rehypothecation of collateral has clear advantages for liquidity in
modern financial systems \citep[see][]{FSB2017}. In particular, it
allows parties to increase the availability of assets to secure their
loans, since a given pool of collateral can be re-used to support
different financial transactions. As a result, rehypothecation
increases the funding liquidity of agents \citep[see][]{brunnermeier2008market}.  At the same time,
rehypothecation also implies risks for market players. First, one risk
associated with the additional funding liquidity allowed by
rehypothecation can be the building-up of excessive leverage in the market
\citep[see
e.g.][]{bottazzi2012securities,singh2012deleveraging,Capel2014collateral}. Second,
rehypothecation implies that several agents are counting on the same
set of collateral to secure their transactions. 
It follows that rehypothecation may represent yet another channel through which
agents' balance sheets become interlocked\footnote{See for
  \citet{battiston2016complexity} for an account of the sources of
  banks' interconnectedness in financial systems.}, and thus a source
of distress propagation and of systemic risk. For instance in the face
of idiosyncratic shocks, some institutions may start to precautionarily
hoard collateral, which in turn constrains the availability of
collateral and its re-use for the downstream institutions in chains of
repledges. This may consequently lead to an inefficient market freeze
if participants lack the necessary assets to secure their loans
\citep{Leitner2011,monnet2011rehypothecation,gorton2012securitized}. The latter is the distress channel
we focus on in this paper. 

To analyze economic benefits and systemic consequences of
rehypothecation, we develop a model of collateral dynamics over a
network of repurchase agreements (repos) across banks. To keep the
model as simple as possible we abstract from many features of markets
with collateral, and we assume that the amount of collateral available
for repo financing is set as a constant fraction of total collateral
available to each agent. The latter includes the proprietary
collateral endowment of each bank as well as the collateral obtained
from other banks via reverse repos. Although simple, our model allows
us to highlight what features of rephypothecation network topology
determine (i) the overall volume of collateral in the market, and (ii)
the velocity of collateral \citep{singh2011velocity}.  We show that
both variables are an increasing function of the length of open
chains.  However, for a given length, cyclic chains (i.e. those where
banks are organized in a closed chain of repo contracts) produce higher
collateral than a-cyclic chains. Furthermore, we show that the
direction of collateral flows also matters. In particular,
concentrating collateral flows among few nodes organized in a cyclic
chain allows large increases in collateral volume and in velocity even
with small chains' length. Finally, we investigate total collateral
under some typical network architectures, which capture different
modes of organization of financial relations in markets, and in
particular different degrees of heterogeneity in the
distribution of repo contracts and of collateral flows.  We show that
total collateral is an increasing function of the density of financial
contracts both in the random network (where heterogeneity is mild) and in
the core-periphery network (where heterogeneity is high).  However,
core-periphery structures allow for a faster increase in collateral. 

The above-explained model with exogenous levels of collateral hoarding rates is useful to analyse
the effects of network topology on collateral flows. At the same time
it is unfit to study the systemic risk implications of rehypothecation
as hoarding behaviour could also reflect the liquidity position of
agents in the network.  We thus extend the basic model, to introduce
hoarding behaviour that accounts for liquidity risk.  More precisely,
we assume that hoarding rates are set according to a Value-at-Risk
(VaR) criterion, aimed at minimizing liquidity default risk. In this
framework, we show that the equilibrium hoarding rate of each bank is
a function of the hoarding rates and the collateral levels of the
banks at which it is directly and indirectly connected. This
introduces important collateral hoarding externalities in the
dynamics, as an increase in hoarding at some banks may indirectly
cause higher hoarding at other banks even not directly
connected to it.  We then use the extended model to study the impact
on total collateral losses of small uncertainty shocks hitting a
fraction of banks in the network, and how those losses vary with the
structure of the rehypothecation networks. We show that core-periphery
structures are the most exposed to large collateral losses when shocks
hit the central nodes in the network, i.e. the one concentrating
collateral flows.  As, core-periphery are also the structures that
generate larger collateral volumes,  our results highlight that
these structures are characterized by a trade-off between liquidity and systemic risk. 

Our work contributes to the recent theoretical literature on the
consequences of collateral rehypothecation \citep[see
e.g.][]{bottazzi2012securities,andolfatto2017rehypothecation,GottardiMaurinMonnet2017,singh2016collateral}. This
literature has highlighted the role of rehypothecation in determining
repo rates \citep[e.g.][]{bottazzi2012securities}, or in softening
borrowing constraints of market participants and in shaping the
interactions in repo markets
\citep{GottardiMaurinMonnet2017,andolfatto2017rehypothecation} or,
finally, it has contributed to evaluate some welfare aspects of policies aimed at
regulating rehypotheaction
\citep{andolfatto2017rehypothecation}. However, to the best of our
knowledge, our paper is the first to study the role of the
structure of the network of collateral exchanges and to explore how
different network structures determine overall collateral
volumes and velocity. Furthermore, our work contributes also to the literature on
liquidity hoarding cascades, and it is particular related to the work
of \citet[][]{gai2011complexity}. However, different from this work,
our model introduces hoarding rates that are responsive to the
liquidity position of the single bank and to the position occupied in
the network. In addition, it shows that liquidity hoarding
dynamics can have quite different consequences depending on the
particular structure of the network.

The paper is organized as follows. Section \ref{sec:model} introduces
the basic definitions used throughout the paper and the model with
fixed hoarding rates. Section \ref{sec:rehyp-netw-endog} studies in
detail how the structure of rehypothecation networks determines
collateral volume and its velocity. Next, Section \ref{sec: Net
  liquidity position, Value-at-Risk, and hoarding effects} extends the
model to feature time-varyng hoarding rates determined according to a
VaR criterion. Section \ref{sec:coll-hoard-casc} uses the latter model
to study collateral hoarding cascades in different rehypothecation
networks. Finally, Section \ref{Conclusions} concludes, also by
discussing some implications of our work.

\section{A model of collateral dynamics on networks}
\label{sec:model}

\noindent In this section we build the network model that we then use
the analyse the ability of the financial system to generate endogenous
collateral in presence of rehypothecation and, next, to study the
dynamics of collateral hoarding cascades in presence of shocks. We
start with basic definitions that we shall use throughout the
paper. We then introduce the laws governing collateral dynamics in
presence of rehypothecation and of fixed hoarding coefficients by
banks. 

\subsection{Definitions}
\label{sec:definitions}

\noindent Consider a set of $N$ financial institutions (``banks'' for
brevity  in the following). Banks invest into an external asset, that
yields an exogenously fixed return $r_{EA}$, and that can also be used
as a collateral. In, addition they lend to each other by using only
secured loans that involve exchange of collateral as in
\citet{singh2011velocity}.\footnote{See also \citet {Aguiar_Map_Collateral2016} for a more comprehensive discussion of the structure of collateral flows.} More precisely, we assume that all debt contracts  are ``repo'' contracts, they are thus secured by collateral. A ``repo'' or  ``repurchase agreement'', is the sale of securities together with an agreement for the seller to buy back the securities at a later date.\footnote{ The repurchase price should be greater than the original sale price, the difference effectively representing interest, and sometimes called the repo rate.} 
A  ``reverse repo'' is the same contract from the point of view of the
buyer. The haircut rate of a repo, that we denote as $h$ , is a
percentage that is subtracted from the market value of an asset that
is being used in a repo transaction.\footnote{The size of the haircut
  usually reflects the perceived risk associated with holding the
  asset.} To collect funds via repo contracts each bank $i$  ($1\leq i \leq N$) can use the collateral that has in its ``box''. The box includes both the proprietary collateral or the collateral obtained via reverse-repos, which can then be re-pledged or rehypothecated for further repo transations.
Repo transactions among banks using proprietary and non-properietary collateral give rise to a directed network $\mathcal {G}$, that we shall label ``rehypothecation network''. To explain in details the dynamics of collateral in each bank's box and the timing of the events occurring through the network  $\mathcal {G}$  it is useful to define the following notations:
\begin{itemize}
\item $A^{C^{out}}_i$: the total amount of collateral flowing out of bank $i$'s box at each step, i.e. the total amount of collateral that the bank $i$ uses to obtain loans from other banks.  
\item $A^{C^{rm}}_i$: the total amount of (re-pledgeable) collateral remaining inside  the box. 
\item $A^C_i$: the total amount of (pledgeable)  collateral flowing into the box of the bank $i$. At every step, the collateral that flows into the box must equal the collateral that remains in the box plus the collateral that flows out of the box. Hence, we have   $A^C_i=A^{C^{out}}_i+A^{C^{rm}}_i$. Notice that $A^C_i$ includes both proprietary as well as  non-proprietary assets received from other banks.  
\item $A^0_i$:  the value of the  proprietary collateral of the bank $i$. This means  the bank $i$ is the original owner of $A^0_i$.
\item $B_i$: the borrowers' set of  bank $i$, i.e. the banks that obtained funding from $i$ via repos and thus provided collateral to $i$.  In the rehypothecation network,  $B_i$ is also the ``in-neighborhood'' of  $i$.
\item $L_i$: the lenders' set of  bank $i$, i.e. the banks  that obtained collateral from $i$ and thus provided funding to $i$. In the rehypothecation network,  $L_i$ is  also the ``out-neighborhood'' of  $i$.
\item Unless specified otherwise, each of the above mentioned symbols written without indices denotes the vector of the same variable for all the banks in the system, for instance  $A^C = [A^C_1, A^C_2, ... A^C_i, ... A^C_N ] $, and similarly for the other quantities.  
\end{itemize}
Furthermore, let the variable $a_{i\leftarrow j}$ capture the direction of collateral flow from the bank $j$ to the bank $i$.  In particular, for every pair of banks $i$ and $j$, $a_{i\leftarrow j}=1$ if bank $j$ has given collateral to bank $i$ and $a_{i\leftarrow j}=0$ otherwise. Two additional variables related to the direction of collateral flows are the ``out-degree'' of a bank $i$, $k_i^{out}$, which measures the total number of outgoing links of the bank, and thus the number of banks to whom bank $i$ provided collateral to. Likewise, the ``in-degree'' of a bank $i$, $k_i^{in}$, is the total number of banks that provided collateral to $i$.

Finally, we assume that each bank hoards a fraction $1-\theta_i$ of the collateral it has in the box. More precisely, for
every monetary unit of collateral, the bank $i$ will keep
$(1-\theta_i)$ inside its box and give away $\theta_i$.  Moreover, to keep the model simple we assume that each bank homogeneously spreads its non-hoarded collateral across its lenders.
Let $s_{i\leftarrow j}$ be the share of bank $j$'s outgoing collateral
flowing into the box of the bank $i$.  If $L_j= \varnothing$
(i.e. $k_j^{out}=0$), then all shares $s_{i\leftarrow j}$ are equal to
zero. If the lender's set is not void, $L_j \neq \varnothing $, that is if $k_j^{out}>0$, then for the total
outgoing collateral pledged or re-pledged by bank $j$, $A^{C^{out}}_j$ it holds:
\[
A^{C^{out}}_j = \sum_{i\in L_j} s_{i\leftarrow j}A^{C^{out}}_j,
\]
since the shares $s_{i\leftarrow j}$ satisfy the constraint $\sum_{i\in L_j} s_{i\leftarrow j}=1$. Notice, that this means that each non-zero column of the matrix of shares  $\mathcal {S}= \{ s_{i\leftarrow j} \}_{N\text{x}N}$ associated with the network $\mathcal {G}$  is summing to 1. In addition, recall that collateral is spread homogenously across lenders.  This implies that
\[
s_j=\frac{1}{k_j^{out}}.
\]
and that the elements of  the matrix $\mathcal {S}$ can be expressed as
\begin{equation}
\begin {cases}
s_{i\leftarrow j}=\frac{a_{i\leftarrow j}}{k_j^{out}}, \ \mbox{if} \ k_j^{out}>0, \\
s_{i\leftarrow j}=0, \ \mbox{otherwise}. \\
\end {cases}
\label{eq_ weights_size_2} 
\end{equation}

\subsection{Collateral dynamics}
\label{sec:coll-dynam-with-const}

To describe collateral dynamics in our model let us assume, in line with \citet{bottazzi2012securities} that the amount of collateral that can be re-hypothecated never exceeds the haircutted amount of collateral. Furthermore, let us assume for simplicity that the haircut rate $h \in [0,1]$ is the same for all banks.  
On these grounds, we can write the following expression for the dynamics of  $A^{C^{out}}_i$, the total amount  collateral  flowing out of  the box of the bank $i$: 
\begin{equation}
A^{C^{out}}_i=A^{0^{out}}_i+(1-h) \delta_i \theta_i  \sum_{j \in B_i} s_{i\leftarrow j}A^{C^{out}}_j,
\label{eq_collateral_box_out}
\end{equation}
where $A^{0^{out}}_i= \delta_i \theta_i A_i^0$ is the proprietary amount of outgoing collateral of the bank $i$. The parameter $\theta_i$ accounts for the fraction that is not hoarded. The second term of the equation captures the amount of collateral received by $i$ from its borrowers $j$ and that is re-pledged. Notice that bank $i$ can only re-pledge a fraction $(1 - h) \theta_i$ of what it receives, where $(1 - h)$ accounts for the fraction remaining after the haircut is applied, and $\theta_i$ accounts for what is not hoarded. Finally, $\delta_i$ is an indicator equal to one if bank $i$ engages in at least one repo contract so that its out-degree is positive and equal to zero otherwise (i.e $\delta_i= 1, \ \mbox {if} \  k_i^{out}>0,$ and zero otherwise). 

Similarly,  the dynamics of the total amount of  re-pledgeable  collateral remaining inside the box of the bank $i$ is described by the following equation 
\begin{equation}
A^{C^{rm}}_i=A^{0^{rm}}_i+(1-h) (1-\delta_i \theta_i)  \sum_{j \in B_i} s_{i\leftarrow j}A^{C^{out}}_j,
\label{eq_collateral_box_rm}
\end{equation}
where $A^{0^{rm}}_i= {A^{0}_i- A^{0^{out}}_i}=(1- \delta_i \theta_i) A_i^0$ is the initial remaining collateral. \\

\noindent Uses and re-use of collateral in our model are fully described by the recursive process explained by the above two equations. Notice that both Equation (\ref{eq_collateral_box_out}) and (\ref{eq_collateral_box_rm}) imply that at the initial step every bank $i$  gives away   $A^{0^{out}}_i$  of collateral to its outgoing neighbors  and keeps $A^{0^{rm}}_i$ inside  its box. In addition, for an amount of $s_{i\leftarrow j}A^{C^{out}}_j$ that the bank  $i$ receives from a neighbour $j$, it re-pledges  $(1-h)\delta_i \theta_i   s_{i\leftarrow j}A^{C^{out}}_j$ and hoards an amount $[1-(1-h)\delta_i \theta_i]  s_{i\leftarrow j}A^{C^{out}}_j$. However, only  the amount $(1-h) (1-\delta_i \theta_i)  s_{i\leftarrow j}A^{C^{out}}_j$ of this hoarded collateral is further re-pledgeable to obtain further funding later, because the haircutted amount of collateral, $h s_{i\leftarrow j}A^{C^{out}}_j$, is kept in a segregated account  that can be only accessed in the case of a credit event \citep[see][for details]{bottazzi2012securities}.


We can also determine the expression of total amount of  re-pledgeable   collateral flowing into the box of each bank  $i$:
\begin{equation}
A^{C}_i= A^{C^{out}}_i+ A^{C^{rm}}_i=A^{0}_i+ (1-h)   \sum_{j \in B_i} s_{i\leftarrow j}A^{C^{out}}_j.
\label{eq_collateral_box_in_1}
\end{equation}
The last equation makes clear that the total collateral flow in the box of a bank,  $A^{C}_i$  includes the proprietary assets (i.e. $A^{0}_i$) as well as  re-pledgeable non-proprietary assets  (i.e. $(1-h)   \sum_{j \in B_i} s_{i\leftarrow j}A^{C^{out}}_j$) received from other banks via reverse repos. Notice that $A^{C^{rm}}_i= (1-\delta_i \theta_i)  A^{C}_i$ and that $A^{C^{out}}_i= \delta_i \theta_i A^{C}_i$. Substituting the latter expression in equation (\ref{eq_collateral_box_in_1}) we obtain a system of equations in the variables $A^{C}$: 
\begin{equation}
A^{C}_i=A^{0}_i+ (1-h) \sum_{j \in B_i} s_{i \leftarrow j} \delta_j \theta_j A^{C}_j.
\label{eq_collateral_box_in_2}
\end{equation}

So far we have not said anything about timing in our model. However, the possibility of collateral use and re-use changes over time as the inflows and outflows of collateral in a bank's box change over time as a consequence of the different uses and re-uses of collateral made by other banks in the network. In addition, the very possibility of re-using collateral is clearly constrained by the maturity of a repo contract $T_{repo}$. In what follows, we shall assume that the maturity of repo contracts is longer than the time scale of the rehypothecation process.\footnote{Notice that in many cases, the rehypothecation process ends already after a small number of steps, typically smaller than the number of banks in the system. In addition, the effective duration of the rehypothecation process exponentially decreases with the levels of the non-hoarding rates and of the haircut rates (see also next section).}

Moreover we shall focus on equilibrium collateral. This equilibrium corresponds to the amount of collateral flow generated by the system over an infinite amount of steps of collateral uses and re-uses. Finally, we shall focus henceforth only on the equilibrium value of the outflowing collateral $A^{C^{out}}$ (equilibrium collateral henceforth). Indeed, first, via Equation \ref{eq_collateral_box_in_1} remaining collateral $A^{C^{rm}}$ is also determined in equilibrium once the amount of outflowing collateral $A^{C^{out}}$ and the initial proprietary collateral $A^{0}$ are known. Second, outflowing collateral is a very interesting variable in our model, as it captures each bank's contribution to overall collateral flows, and thus to overall funding liquidity in the market. 

To find the equilibrium of outflowing collateral, let us start by writing equation \eqref{eq_collateral_box_out} in matrix form
\begin{equation}
A^{C^{out}}=A^{0^{out}}+(1-h)\mathcal {M}A^{C^{out}}
\label{Eq:outgoing_collateral_matrix}
\end{equation}
where the the elements $ \mathcal {M}$ is the adjacency matrix of the rehypothecation network  $\mathcal {G}$, with elements $m_{i\leftarrow j}$ defined as:
\[
\begin{cases}
m_{i\leftarrow j}=\delta_i \theta_i  s_{i\leftarrow j} =\frac{\delta_i \theta_i a_{i\leftarrow j}}{k_j^{out}}, \ \mbox {if} \  k_j^{out}>0,\\ 
m_{i\leftarrow j}=0, \ \mbox {if} \  k_j^{out}=0. \\ 
\end{cases}
\]
Given the network of collateral flows $\mathcal {G}$, the haircut rate $h$, and the vector of non-hoarding rates $\theta= \{ \theta_i \}_{i=1}^{n}$, we can obtain the equilibrium value of $A^{C^{out}}$ by solving equation  (\ref{Eq:outgoing_collateral_matrix}) as follows:
\begin{equation}
A^{C^{out}}=(\mathcal {I}-(1-h) \mathcal {M})^{-1}A^{0^{out}}=\mathcal {B}_1 A^{0^{out}},
\label{eq_tot_outflow_collateral}
\end{equation}
with $\mathcal {I}$ is the identity matrix of size $N$,  $\mathcal {B}_1=(\mathcal {I}-(1-h)  \mathcal {M})^{-1}$. The above equation indicates that equilibrium collateral will in general be a function of the 
of the entire topology of the rehypothecation network $\mathcal {G}$ and of the vector of non-hoarding rates $\theta= \{ \theta_i \}_{i=1}^{n}$. In the next sections we first study the role of the network topology in affecting collateral flows and in determining different levels of equilibrium collateral. 

\section{Rehypothecation networks and endogenous collateral}
\label{sec:rehyp-netw-endog}

We shall now describe how the structure of the rehypothecation network affects collateral flows and equilibrium collateral determined according to the model developed in the previous section. To perform our investigation it is useful to define some aggregate indicators mesuring the performance of a network in affecting collateral flows. The first one is the aggregate amount  of outgoing collateral, $S^{out}$, or ``total collateral'' henceforth, which is defined as: 
\begin{equation}
S^{out}=\sum_{i=1}^{i=N} A^{C^{out}}_i,
\label{eq_aggregate_outflow}
\end{equation}
In addition, we also introduce the multiplier of the aggregate amount of proprietary collateral\footnote {See \cite {FSB2017_measure} for discussions of other collateral re-use measures.}, $m$, or ``collateral multiplier'', henceforth. It is defined as:
\begin{equation}
m= \frac {\sum_{i=1}^{i=N} A^{C^{out}}_i} {\sum_{i=1}^{i=N} A^{0^{out}}_i} = \frac {S^{out}}{S^{0^{out}}}, 
\label{eq_multiplier_outflow}
\end{equation}
where $S^{0^{out}}=\sum_{i=1}^{i=N} A^{0^{out}}_i$ is the total outflowing proprietary collateral. Throughout the paper,  we shall focus on   $S^{out}$ and  $m$ when analyzing collateral creation allowed by a given network $\mathcal{G}$. Notice that  $S^{out}$ captures the aggregate flow of collateral provided by banks in the financial system. A higher (lower) amount of this flow indicates a more liquid market, i.e. one where agents can easily find collateral to secure their financial transactions. Furthermore, notice that  the denominator on the right hand side of equation (\ref{eq_multiplier_outflow}) is the aggregate amount of initial ouflowing collateral. Accordingly, $m$ captures the velocity of collateral when rehypothecation is allowed \citep[see also][]{singh2011velocity}. Again, a higher value of $m$ indicates a more liquid market, and in particular one where the same set of collateral can secure a larger set of secured lending contracts. In that respect, we shall also say that a rehypothecation network $\mathcal{G}$ generates ``endogenous collateral'' whenever the collateral multiplier associated with it is larger ($m>1$)
 Clearly, both performance indicator are affected by the hoarding behavior of banks in the network. To simplify the analysis in this section we shall assume the non-hoarding and hoarding rates are fixed, so that $\theta_i=\bar{\theta}$ and $(1-\theta_i)=(1-\bar{\theta})$, $\forall i$. In Section 4, we shall remove this restriction, and we shall discuss how banks set these coefficients endogenously, according to a VaR criterion in presence of liquidity shocks. 


We shall begin our analysis by providing stylized examples and by stating proposition that show how the aggregate amount of collateral going out of the boxes of all banks, $S^{out}$, and the multiplier of collateral, $m$, are influenced by some key characteristics of rehypothecation networks, like the length of rehypothecation chains, the presence of cyclic chains, or the direction of collateral flows. In additions, these results shed light on the core mechanisms driving the generation of endogenous collateral and of collateral hoarding cascades in more complex network architectures.



\subsection{Length of chains and network cycles}
\label{sec:length-chains-cycles}
Let us start with simple examples of chains composed by three banks as
shown in Figure \ref {fig_three_nodes}:  panel (a) a star chain, panel
(b) an open chain or ``a-cyclic'' chain, and panel (c) a closed chain or a ``cycle''.

\begin{figure}[H]
\centering
\captionsetup[subfloat]{farskip=0pt,captionskip=0pt}
\includegraphics[width=15cm]{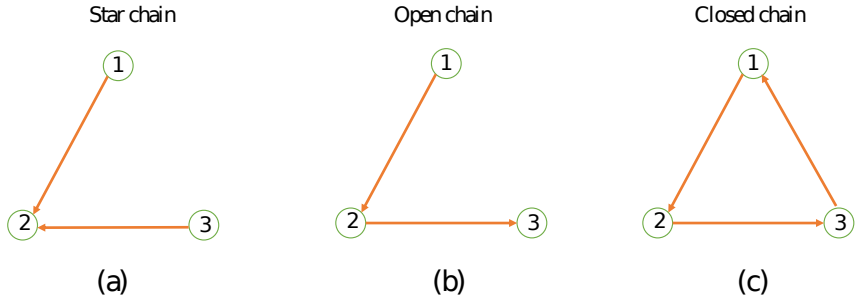}
   \vspace*{-3.5cm}
\caption{Examples of rehypothecation chains among three banks.} 
\label{fig_three_nodes}
\end{figure}

\bigskip

\textit {Star chains} 

\bigskip

In the first case (i.e. the star chain in  Figure \ref
{fig_three_nodes} (a)),  B2 receives collateral from  B1 and B3, and
then it does not re-use it.  We will show that without
rehypothecation, there is no endogenous collateral created in the
system. Before discussing the example,  we denote by  $A^{C^{out}}_{i, T}$  the cumulated amount of collateral ougoing from the box of the bank $i$ after $T$ times. In addition, $ S_{a, T}^{out}$ indicates total collateral after $T$ times and $ m_{a, T}$ the corresponding multiplier.\footnote {Notice that $A^{C^{out}}_{i}$  mentioned in equation (\ref{eq_tot_outflow_collateral}) is equilibrium collateral, and it corresponds to the amount of outflowing  collateral when $T \to \infty$.}
At $t = 1$, the initial  total  amounts of collateral  outgoing from the boxes of   B1, B2, B3 are
\begin{equation}
\begin{cases}
A_{1,t = 1}^{C^{out}}=A^{0^{out}}_1=  \theta_1A_{1}^{0}  \ \mbox {(going to bank $2$)}  \\  
A_{2,t = 1}^{C^{out}}=A^{0^{out}}_2=0\\ 
A_{3,t = 1}^{C^{out}}=A^{0^{out}}_3=\theta_3A_{3}^{0} \ \mbox {(going to bank $2$)}  \\ 
\end{cases}
\end{equation}
At $t= T \geq 2$,  $A^{C^{out}}_i$ remains constant   $\forall i$ since there is no re-use of collateral. We can write
\begin{equation}
\colvec[A^{C^{out}}_{1,t= T}]{A^{C^{out}}_{3,t= T}}{A^{C^{out}}_{3,t= T}}= \colvec[A^{0^{out}}_1]{A^{0^{out}}_2}{A^{0^{out}}_3} + (1-h) M_a
 \colvec[A^{0^{out}}_1]{A^{0^{out}}_2}{A^{0^{out}}_3},
\end{equation}
where  \[M_a=\begin{bmatrix}
    0       &  0 &   0  \\
0           & 0 & 0  \\
   0      & 0  & 0
\end{bmatrix}.\]  
It follows that total collateral in the example of Figure \ref {fig_three_nodes} (a) is always equal to the sum of the initial prorietary collateral outflowing from banks' boxes and thus, that there is no creation of endogenous collateral. That is, we get:  
\begin{equation}
S_{a}^{out}= S_{a, T}^{out}= \sum_{i=1}^{i=3} A^{C^{out}}_{i, T}= A^{0^{out}}_1 + A^{0^{out}}_3,
\end{equation}
and
\begin{equation}
 m_{a}= m_{a, T}= \frac {\sum_{i=1}^{i=3} A^{C^{out}}_{i,  T}}{\sum_{i=1}^{i=3} A^{0^{out}}_i} =1.
\end{equation}

\bigskip

\textit {A-cyclic chains} 

\bigskip

We now consider  the second example represented by the open chain or a-cyclic chain  in  Figure \ref {fig_three_nodes} (b). In this case  B2 can re-use the collateral that it receives from  B1.  In presence of rehypothecation the network generates endogenous collateral, $S^{out}>S^{0^{out}}$. However, the possibilities of endogenous collateral creation are constrained by the length of the open chain (and equal to 2 in the example shown in the figure).
At $t = 1$, the initial amounts of collateral  outgoing from the boxes of  B1, B2, B3 are
\begin{equation}
\begin{cases}
A_{1,t = 1}^{C^{out}}=A^{0^{out}}_1 =  \theta_1A_{1}^{0}  \ \mbox {(going to bank $2$)}  \\ 
A_{2,t = 1}^{C^{out}}=A^{0^{out}}_2=\theta_2 A_2^0  \ \mbox {(going to bank $3$)}  \\ 
A_{3,t = 1}^{C^{out}}=A^{0^{out}}_3=0\\ 
\end{cases}
\end{equation}
At $t =  2$,    bank $2$ will re-use a fraction $\theta_2$ of an additional  (re-pledgeable)  collateral
 that it has received from the bank $1$ at time $t=1$. Therefore, the amounts of collateral outgoing from the boxes of   B1, B2, B3 are
\begin{equation}
\begin{cases}
A_{1,t = 2}^{C^{out}}=  A^{0^{out}}_1\\ 
A_{2,t = 2}^{C^{out}}= A^{0^{out}}_2 +  (1-h)    \theta_2 A^{0^{out}}_1\\ 
A_{3,t = 2}^{C^{out}}=  A^{0^{out}}_3\\ 
\end{cases}
\end{equation}
which in matrix form reads
\begin{equation}
\colvec[A^{C^{out}}_{1,t = 2}]{A^{C^{out}}_{3,t = 2}}{A^{C^{out}}_{3,t = 2}}= \colvec[A^{0^{out}}_1]{A^{0^{out}}_2}{A^{0^{out}}_3} + (1-h) M_b
 \colvec[A^{0^{out}}_1]{A^{0^{out}}_2}{A^{0^{out}}_3}, 
\end{equation}
where  now \[M_b=\begin{bmatrix}
    0       &  0 &   0  \\
\theta_2           & 0 & 0  \\
   0      & 0  & 0
\end{bmatrix}.\] 
Since  all elements of $M_b^{t}$ are equal to zero for all $t \geq 2$, we get that the equilibrium values of total outflowing collateral and of the corresponding multiplier are:
 \[ \colvec[A^{C^{out}}_{1,t = T}]{A^{C^{out}}_{2,t = T}}{A^{C^{out}}_{3,t = T}}=\colvec[A^{C^{out}}_{1,t = 2}]{A^{C^{out}}_{3,t = 2}}{A^{C^{out}}_{3,t = 2}} (\forall T \geq 2). \] 
In addition,
\begin{equation}
S_{b}^{out}= S_{b,T}^{out}= \sum_{i=1}^{i=3} A^{C^{out}}_{i, t=T}= A^{0^{out}}_1 + A^{0^{out}}_2 +  \theta_2  (1-h)  A^{0^{out}}_1,
\end{equation}
and
\begin{equation}
 m_{b}=  m_{b,T}= \frac {\sum_{i=1}^{i=3} A^{C^{out}}_{i, t=T}}{\sum_{i=1}^{i=3} A^{0^{out}}_i} =1 +  \theta_2  (1-h)  \frac {A^{0^{out}}_1} {A^{0^{out}}_1 + A^{0^{out}}_2}.
\end{equation}
Notice that the above collateral multiplier is larger than 1 as long as $h<1$ and $\theta_2>0$.

\bigskip

\textit{Cyclic chains}

\bigskip

We now consider the third case when rehypothecation processes among  banks create a closed chain or a ``cycle'', like the one in  Figure \ref {fig_three_nodes} (c).
Notice that in the above example every bank has a positive out-degree, i.e. $k^{out}_i >0, \forall i=1, 2, 3$ and accordingly, $\delta_i =1, \forall i=1, 2, 3$ (cf. Section \ref{sec:coll-dynam-with-const}). We will now show that the creation of endogenous collateral is no longer constrained by the length of the chain, and thus that in the end total collateral and the multipliers are larger than in previous example.

At $t=1$,   the initial total  amounts of outgoing collateral are:
\begin{equation}
\begin{cases}
A^{C^{out}}_{1,t=1}=A^{0^{out}}_1  = \theta_1 A^{0}_{1} \ \mbox {(going to bank $2$)}  \\ 
A^{C^{out}}_{2,t=1} =A^{0^{out}}_2 =  \theta_2 A^{0}_{2} \ \mbox {(going to bank $3$)}  \\ 
A^{C^{out}}_{3,t=1} =A^{0^{out}}_3  =\theta_3 A^{0}_{3}\ \mbox {(going to bank $1$)} \\ 
\end{cases}
\end{equation}
Furthermore, at $t=2$,  each bank $i$ will re-use a fraction $\theta_i$ of the additional  re-pledgeable collateral that it has received from other banks at the previous time. We thus get:
\begin{equation}
\begin{cases}
A^{C^{out}}_{1,t=2}= A^{0^{out}}_1+  \theta_1 (1-h)  A^{0^{out}}_3   \\ 
A^{C^{out}}_{2,t=2}= A^{0^{out}}_2 +  \theta_2  (1-h)  A^{0^{out}}_1  \\ 
A^{C^{out}}_{3,t=2}=A^{0^{out}}_3 +    \theta_3  (1-h)  A^{0^{out}}_2 \\ 
\end{cases}
\end{equation}
and in matrix form 
\begin{equation}
\colvec[A^{C^{out}}_{1,2}]{A^{C^{out}}_{3,2}}{A^{C^{out}}_{3,2}}= \colvec[A^{0^{out}}_1]{A^{0^{out}}_2}{A^{0^{out}}_3} + (1-h)M_c
 \colvec[A^{0^{out}}_1]{A^{0^{out}}_2}{A^{0^{out}}_3},
\end{equation}
where \[M_c=\begin{bmatrix}
    0       &  0 &   \theta_1  \\
\theta_2           & 0 & 0  \\
   0      & \theta_3  & 0
\end{bmatrix}.\]
Notice that $\theta_1 (1-h) A^{0^{out}}_3$, $ \theta_2  (1-h) A^{0^{out}}_1$,  $ \theta_3 (1-h)  A^{0^{out}}_2$ are respectively the additional amounts of collateral that banks $1$, $2$, and $3$ receive from all other banks. Moreover, at $t=3$,  each bank $i$ will again re-use a fraction $\theta_i$ of the additional   re-pledgeable collateral that it has received from other banks at time $t=2$. Therefore,
\begin{equation}
\begin{cases}
A^{C^{out}}_{1,t=3}= A^{0^{out}}_1+  \theta_1 A^{0^{out}}_3 (1-h) + \theta_1  (1-h)  \theta_3 (1-h) A^{0^{out}}_2\\ 
A^{C^{out}}_{2,t=3}= A^{0^{out}}_2 +   \theta_2  A^{0^{out}}_1 (1-h)  + \theta_2  (1-h)  \theta_1 (1-h) A^{0^{out}}_3 \\ 
A^{C^{out}}_{3,t=3}=A^{0^{out}}_3 +     \theta_3 A^{0^{out}}_2 (1-h) + \theta_3  (1-h) \theta_2  (1-h) A^{0^{out}}_1\\ 
\end{cases}
\label{eq:closed_chain_system}
\end{equation}
which in matrix form reads
 \begin{equation}
\colvec[A^{C^{out}}_{1,3}]{A^{C^{out}}_{3,3}}{A^{C^{out}}_{3,3}}= \colvec[A^{0^{out}}_1]{A^{0^{out}}_2}{A^{0^{out}}_3} + [(1-h) M_c]^1
 \colvec[A^{0^{out}}_1]{A^{0^{out}}_2}{A^{0^{out}}_3} + [(1-h) M_c]^2
 \colvec[A^{0^{out}}_1]{A^{0^{out}}_2}{A^{0^{out}}_3}.
\end{equation}
In general, at $t=T+1$, i.e. after $T$ times of collateral re-uses, the cumulated amounts of outgoing collateral are:
\begin{equation}
\colvec[A^{C^{out}}_{1,T+1}]{A^{C^{out}}_{2,T+1}}{A^{C^{out}}_{3,T+1}}= \{ I+[(1-h) M_c]^1+ [(1-h) M_c]^2+...+[(1-h) M_c]^T \}
\colvec[A^{0^{out}}_1]{A^{0^{out}}_2}{A^{0^{out}}_3}.
\end{equation}
Expressed differently,
\begin{equation}
\colvec[A^{C^{out}}_{1,T+1}]{A^{C^{out}}_{2,T+1}}{A^{C^{out}}_{3,T+1}}= \colvec[A^{0^{out}}_1]{A^{0^{out}}_2}{A^{0^{out}}_3} + (1-h) M_c \colvec[A^{C^{out}}_{1,T}]{A^{C^{out}}_{2,T}}{A^{C^{out}}_{3,T}}.
\end{equation}
Clearly, the additional collateral in the system is now equal to
\[[((1-h) M_c)^1+ ((1-h) M_c)^2+.....((1-h) M_c)^T
]\colvec[A^{0^{out}}_1]{A^{0^{out}}_2}{A^{0^{out}}_3}.\] 

Finally, when $T \to \infty$, we obtain the equilibrium values for $S_{c}^{out}$  and for $m_{c}$. Their expressions are the following:
\begin{equation}
S_{c}^{out}= \lim_{t \to\infty} S_{c, t}^{out}=\lim_{t \to\infty} \sum_{i=1}^{i=3} A^{C^{out}}_{i, t},
\end{equation}
and
\begin{equation}
 m_{c}=  \lim_{t \to\infty} m_{c, t}= \lim_{t \to\infty} \frac {\sum_{i=1}^{i=3} A^{C^{out}}_{i, t}}{\sum_{i=1}^{i=3} A^{0^{out}}_i}.
\end{equation}
To conclude, it is interesting to notice that, as long as $\{ A_i^0 \}_{i=1}^{i=3}$ and $\{ \theta_i \}_{i=1}^{i=3}$ are homogenous across banks, we have the following ranking  $S_c^{out} >S_b^{out}> S_a^{out}$ and $m_c >m_b> m_a$.

\subsection{Direction of collateral flows, collateral sinks and cycles' length} 
\label{sec:the role of the direction of collateral flows}

The above examples have clarified what are the fundamental properties that a rehypothecation network must have in order to create endogenous collateral, and thus additional liquidity in the system. In particular, the second example (the a-cyclic chain) makes clear that the possibilities of additional collateral creation are determined by the length of the repledging chains among banks. However, the presence of cycles in networks, like the third example above, allows one to go beyond that, and to maximize the amount of collateral creation. Furthermore, in presence of cyclic networks the direction of collateral flows also matters. In particular, networks wherein collateral flows all end up in a cycle will ceteris paribus create more endogenous collateral than networks where some collateral leaks out from cycles and sinks in some nodes of the system.  To better clarify the foregoing statement, we consider in Figure \ref {fig_five_nodes}, three example networks of five nodes with the same number of links. In addition, in these networks, all out-degrees are positive, and consequently collateral from each node will flow into a cycle after going through some directed edges\footnote{It is also possible to show that, as long as all agents have positive out-degree, collateral flows will end up in a cycle after a finite number of steps. For the sake of brevity we do not report this proposition and the related proof. However, it is available from the authors upon request.}, and there is no leakage from the cycle. The only difference among the three networks is in the length of cycles. Nevertheless, we will show that as long as non-hoarding rates are constant and homogeneous, the three networks generate the same equilibrium total collateral and the have the same equilibrium multiplier. This is summarized in the following proposition.

\begin{proposition}
Let $\theta_i =\theta \ (\forall i)$. For the networks in panels  $\alpha=a, b, c$ of Figure (\ref{fig_five_nodes}) $S_{\alpha, t}^{out}= \sum_{i=1}^{i=5} A^{C^{out}}_{i, t}, \;\;\ \forall t\geq 1$.  
\label{prop:independence-of-cycles-length} 
\end{proposition}
\begin{proof}
 See appendix. 
\end{proof}

\begin{figure}[H]
\centering
\captionsetup[subfloat]{farskip=0pt,captionskip=0pt}
\includegraphics[width=15cm]{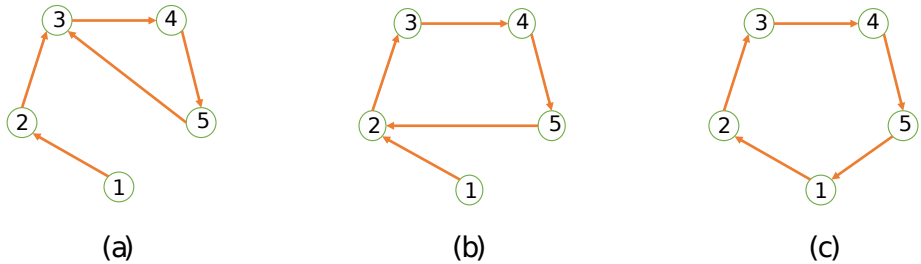}
   \vspace*{-4cm}
\caption{Different rehypothecation process among five banks: same number of links, different lengths of cycles but same total amounts of equilibrium collateral. }
\label{fig_five_nodes}
\end{figure}

Let us next consider three other examples of rehypothecation process among five banks, the ones in Figure (\ref{fig_five_nodes_2}). The difference with examples of Figure \ref{fig_five_nodes} is that now in panels (a) and (b) one node (i.e.  node 5) has zero out-degree. In addition, the network structure in these two panels imply that some collateral leaks out of a cycle and gets stuck at the node with zero out-degree, which then plays the role of ``collateral sink''. The consequence is that the three networks will generate different amounts of total equilibrium collateral. This is stated in the following proposition.
\begin{proposition}
  Let $\theta_i =\theta \ (\forall i)$. For the networks in panels  $\alpha=a, b, c$ of Figure \ref{fig_five_nodes_2},  $S_{a, t}^{out} < S_{b, t}^{out} < S_{c, t}^{out}, \;\;\, \forall t\geq 2.$
 \label{prop:dependence-cycles-length}
\end{proposition}
\begin{proof}
  See appendix.
\end{proof}

\begin{figure}[H]
\centering
\captionsetup[subfloat]{farskip=0pt,captionskip=0pt}
\includegraphics[width=15cm]{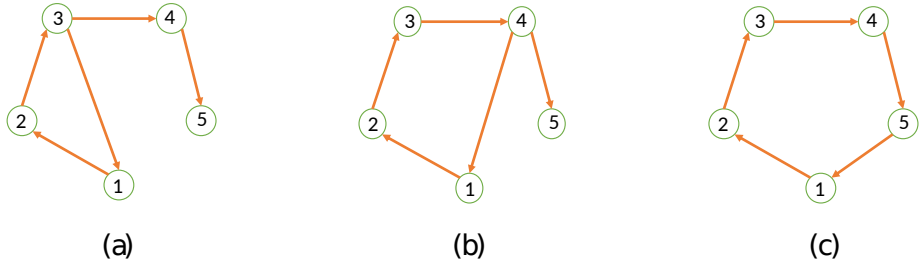}
   \vspace*{-4cm}
\caption{Different rehypothecation process among five banks: same number of links but different total amounts of equilibrium collateral. }
\label{fig_five_nodes_2}
\end{figure}

The above two propositions deliver interesting implications about the
role of networks' topology in determining total collateral
flows. First, Proposition \ref{prop:dependence-cycles-length} shows
that the total amount of collateral is maximized when the longest
possible cycle in the network has been created (a cycle of length 5 in
example (c) of Figure \ref{fig_five_nodes_2}). At the same time,
proposition \ref{prop:independence-of-cycles-length} shows that
cycles' length is irrelevant when collateral sinks are not present in
the network and all banks have positive out-degree, i.e. they have at
least one repo with some other bank in the system. It also follows
that in that case, it might be advantageous frow the viewpoint of
total creation of collateral in the network to concentrate collateral
flows among few nodes that have a cyclic chain among them. These
results provide key insights to understand the behaviour of collateral
flows in different network architectures that we shall discuss in the
next section. They are also central to understand some of the results
about collateral hoarding cascades that we will expose in Section
\ref{sec:coll-hoard-casc}. 

The next two propositions generalize the above results to any network architecture. The first of them shows that adding an arbitrary number of links to a cycle of length equal to the size of the network does not change neither total collateral nor the value of the multiplier. The second one identifies the upper bounds for the equilibrium values of total collateral and of the collateral multiplier and shows that these upper limits are attained as long as every bank has at least one outgoing link in the network. 

\begin{proposition}
Consider a rehypothecation network $\mathcal{G}$ of size $N$. Let $\theta_i =\theta$ and
$A_i^0=A^0  \;\ (\forall i)$. Adding arbitrary links to an initial cycle of size $N$ does not change  the values of $S^{out}$ and $m$.  More in general, as long as there is the presence of the largest cycle,  $S^{out}$ and $m$ remain unchanged as the density inside that cycle increases.
\label{prop:proposition_k_reg}
\end {proposition}
\begin{proof}
  See the appendix.
\end{proof}

\begin{proposition}
Consider a rehypothecation network $\mathcal{G}$ of size $N$. Let $\theta_i =\theta$ and
$A_i^0=A^0  \;\ (\forall i)$.  If   $k_i^{out} >0  \;\ (\forall i)$ then equilibrium values of $S^{out}$ and $m$ are equal to the following upper limits:
\[S^{out}= \frac  {\theta}{1-(1-h)\theta}  N A^0,\]
and
\[m =\frac  {1}{1-(1-h)\theta}.\] 
\label{prop:proposition_positive_out_degree_basic}
\end {proposition}
\begin{proof}
  See the appendix.
\end{proof}

\subsection {Network architecture and collateral creation}
\label{sec:netw-arch-coll}

\noindent We now address the issue of how does the network structure more in general affects collateral creation. We shall focus on three very different classes of network structures. These classes represent general archetypes of networks in the literature, and they also capture some idealized modes of organization of financial contracts in the market. The first class consists of the closed k-regular graphs of size $N$, $\mathcal{G}_{reg}$ , wherein each node has $k$ in-coming neighbors as well as $k$ out-going neighbors. This archetype corresponds to a market where repo contracts are homogenously spread across banks, so that each bank has exactly the same number of repos and of reverse repos. We consider different types of closed regular graphs of varying levels of density $\frac{1}{N-1}<p<1$. Special cases of this structure are the cycle  of size $N$ (where $p=\frac{1}{N-1}$) and the complete network ($p=1$), wherein each bank has a repo with every other bank in the network (and vice-versa) and where the number of incoming and outgoing links is the same for all banks and equal to $N-1$. The second class consists of the random graphs $\mathcal{G}_{rg}$, in which there is a mild degree of heterogeneity in the distribution of financial contracts across banks. Here, the probability of a directed link (and thus of the existence of repo) between every two nodes is equal to the density $p$ ($0\leq  p \leq 1$) of the network. Notice that as $p\to1$ the random graph converges to the complete graph. Finally, the third class we examine here consists of core-periphery networks, $\mathcal{G}_{cp}$, where (i)  the number of nodes in the core, $N_{core}$,  is fixed; (ii) each node in the periphery has only out-going links, and all point to nodes in the core; (iii) nodes in the core are also randomly connected among themselves with the probability $p_{core}$ ($0\leq  p_{core} \leq 1$), and there are no directed links from the core to the periphery nodes. Notice that the latter type of structure exacerbates heterogeneity in the distribution of financial contracts and it centralizes collateral flows among nodes in the core. This structure is also interesting from an empirical viewpoint as high concentration of collateral flows is often observed in actual markets \citep[see e.g.][]{singh2011velocity}. 

We begin our analysis of the three structures by characterizing the behaviour of endogenous collateral formation in the closed k-regular graphs.

\begin{proposition}
Let $\theta_i =\theta$ and $A_i^0=A^0  \;\ (\forall i)$. A closed-k regular rehypothecation network $\mathcal{G}_{reg}$ of size $N$, and a density $p=\frac{k}{N-1}$, always returns the same equilibrium values of total collateral $S_{reg}^{out}$ and of collateral multiplier $m_{reg}$ for any density $\frac{1}{N-1}<p<1$. The equilibrium values are given by the limits stated in Proposition \ref{prop:proposition_positive_out_degree_basic}. 
\label{prop:new_prop_k_reg}
\end{proposition}

The above statement follows directly from Proposition \ref{prop:proposition_k_reg} above. A closed k-regular of density $\frac{1}{N-1}$ already embeds the longest possible cycle that is possible to create in a network of size $N$. Accordingly adding further links does not  bring any change in equilibrium values of total collateral and of the multipliers, which are always equal to the upper bounds stated in Proposition \ref{prop:proposition_positive_out_degree_basic}. In contrast to the close k-regular graph, the random network and of core-periphery networks display some variation in total collateral and in the multiplier with for increasing levels of density. The following proposition characterizes the behavior of the latter two variables in these two network structures.\footnote{In the next proposition we determine the equilibrium for bank's collaterals in the case of an average system, that is instead of considering each single sample of collateral networks we consider the expected value of the amount of collateral for each banks, and solve only for a single average system. For all cases, numerical evidence strongly supports the analytical results.} 
\begin{proposition}
Let $\theta_i =\theta$ and $A_i^0=A^0  \;\ (\forall i)$, then: 
\begin{enumerate}
\item A random graph $\mathcal{G}_{rg}$ of size $N$ creates more equilibrium collateral and a higher multiplier with higher level of  density $p$.
\item A core-periphery graph of $\mathcal{G}_{cp}$ of size $N$ creates  more equilibrium collateral and a higher multiplier with higher level of density $p_{core}$ of the core.
\item For any graph density $0\leq  p \leq 1$, $\exists \, p_{th}=p_{th}(N,N_{core},p)$ such that a core-periphery graph of $\mathcal{G}_{cp}$ creates more equilibrium collateral and a higher multiplier than random graphs $\mathcal{G}_{rg}$ for any $p_{core}>p_{th}$.
\item As the overall density $p$ in the random graph (or  the density  $p_{core}$ of the core in the core-periphery  graph) goes to 1, the equilibrium values of total collateral and the multiplier converge to the limits stated in Proposition \ref{prop:proposition_positive_out_degree_basic}.
\end{enumerate}
\label{prop:proposition_density_rand} 
\end {proposition} 
\begin{proof}
  See the appendix.
\end{proof}

\begin{figure}[H]
\centering
\captionsetup[subfloat]{farskip=0pt,captionskip=0pt}
\subfloat[] {\includegraphics [width=5cm,  height=4cm]{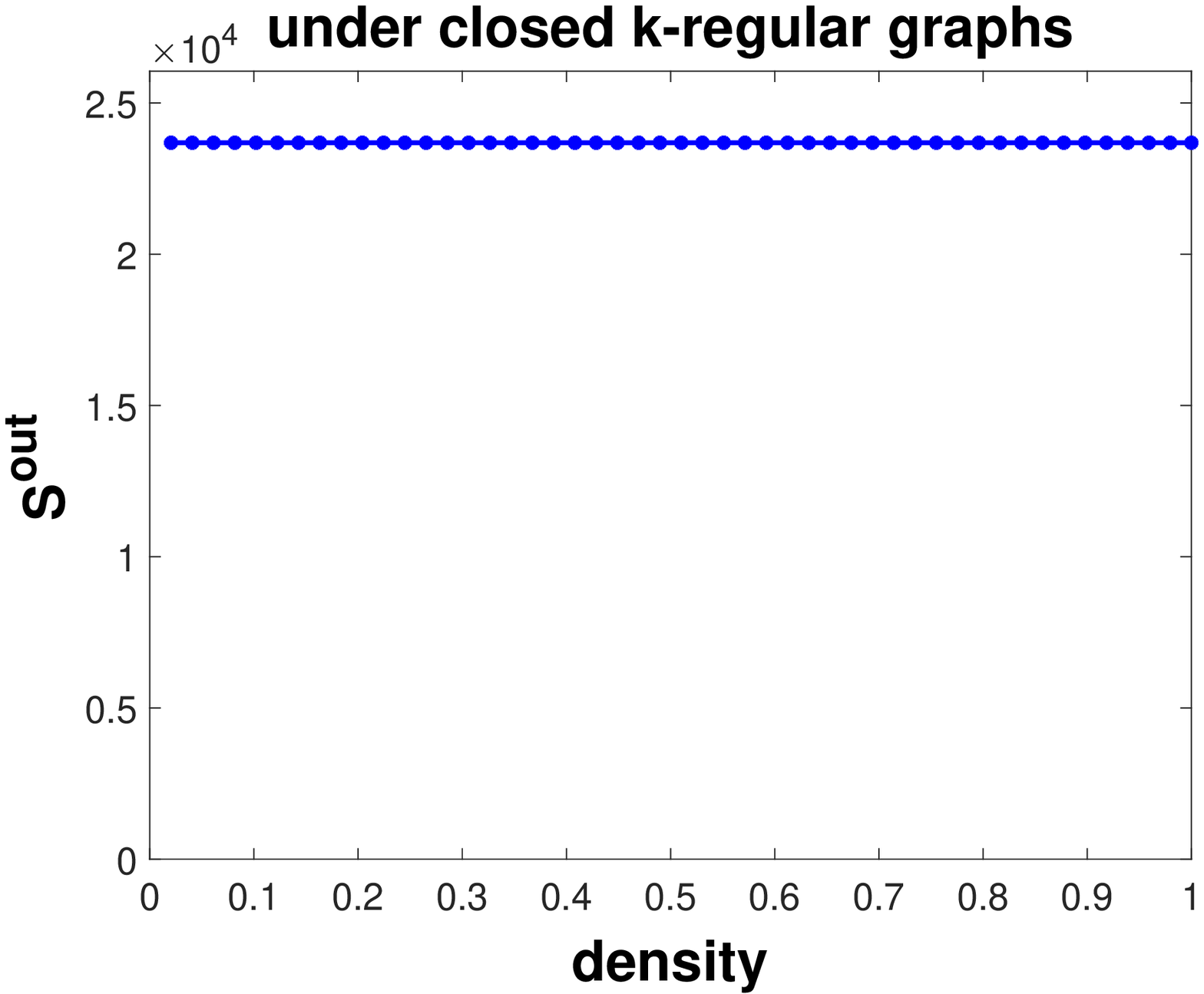}}
\subfloat[] {\includegraphics [width=5cm,  height=4cm]{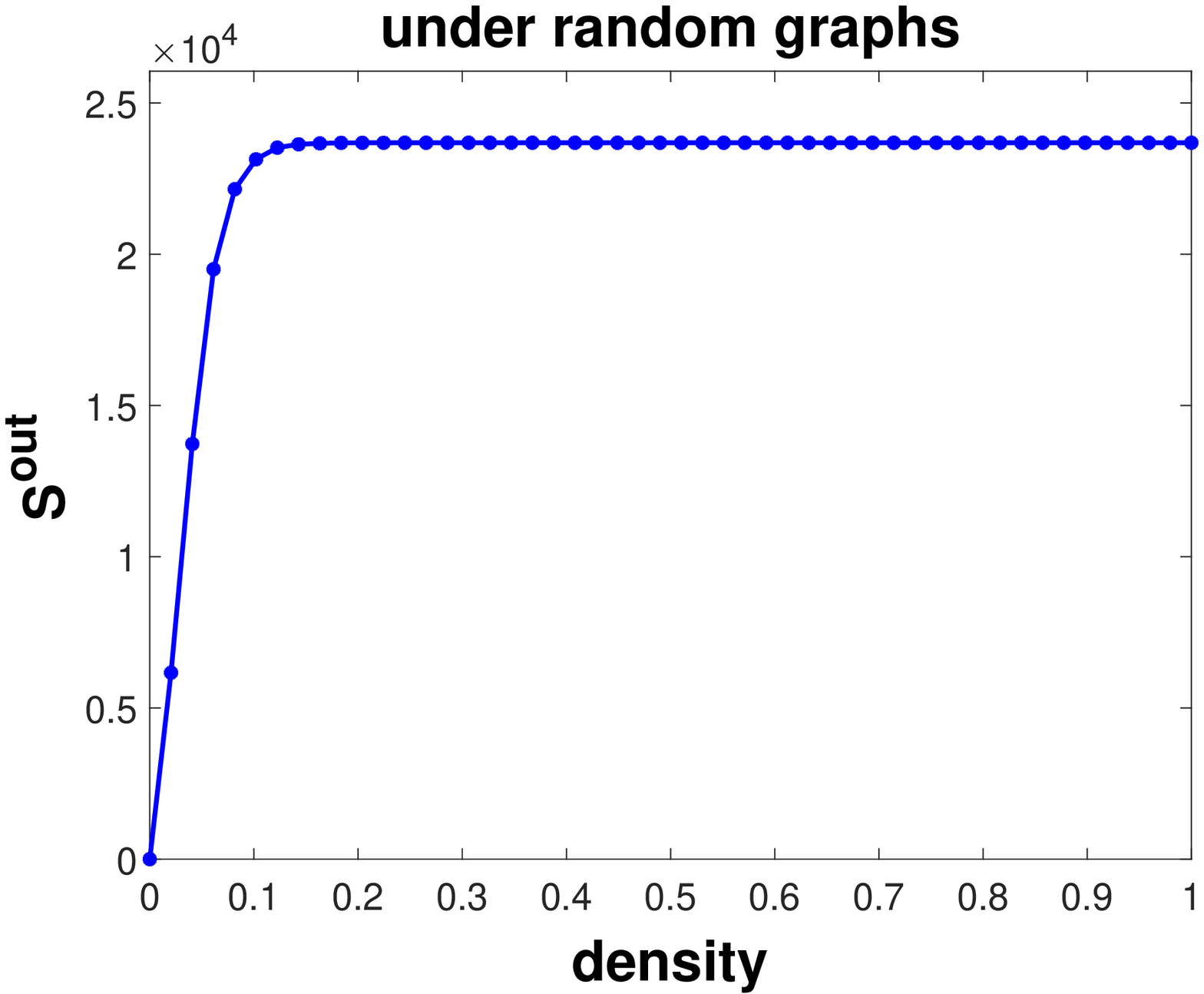}}
\subfloat[] {\includegraphics [width=5cm,  height=4cm]{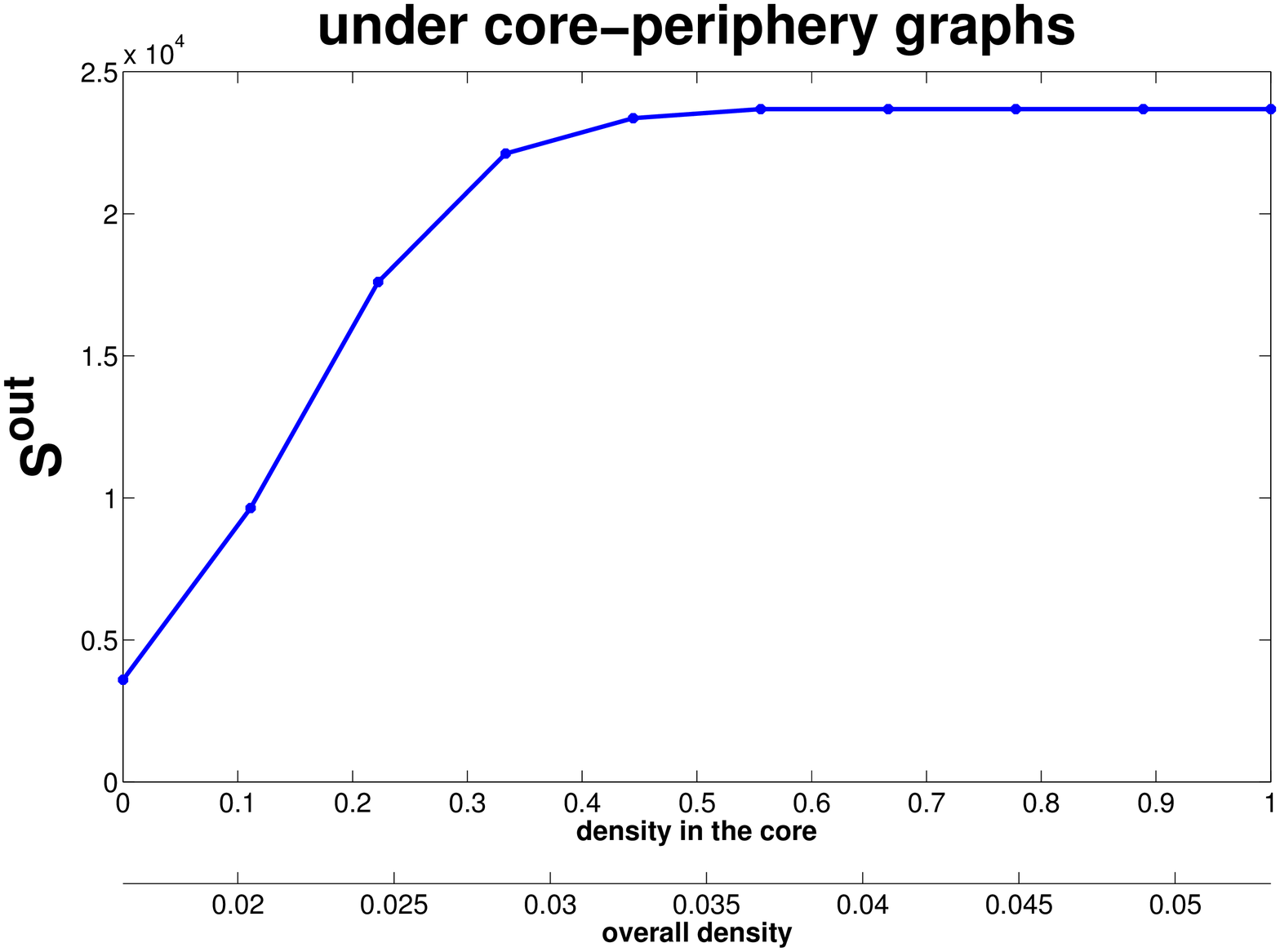}}
\caption{Total collateral ($S^{out}$) as a function of density under different network structures. Notice the different scale for the overal density in panel (c).}
\label{fig:collateral_outflowing_example}
\end{figure}

\begin{figure}[H]
\centering
\captionsetup[subfloat]{farskip=0pt,captionskip=0pt}
\subfloat[] {\includegraphics [width=5cm,  height=4cm]{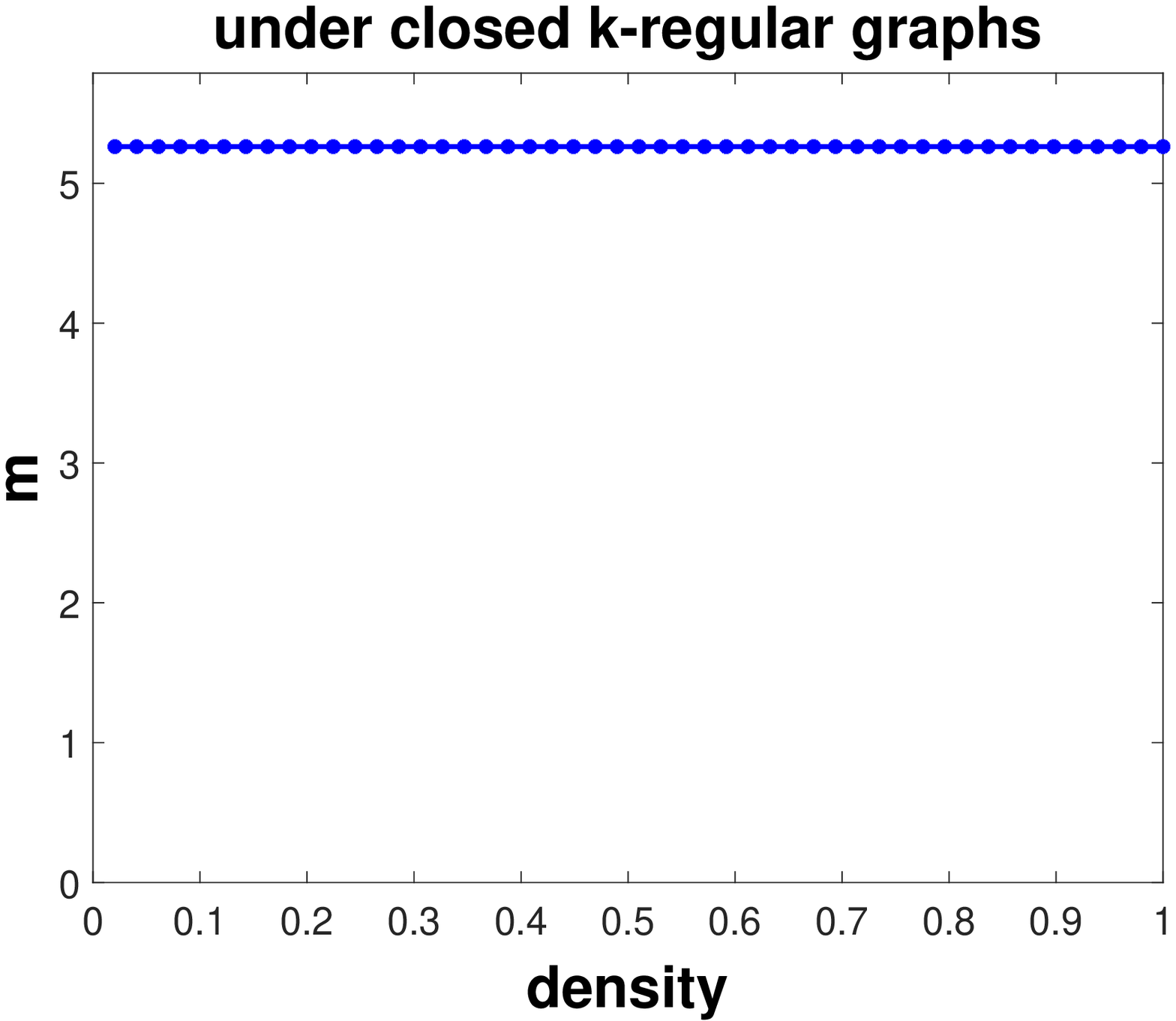}}
\subfloat[] {\includegraphics [width=5cm,  height=4cm]{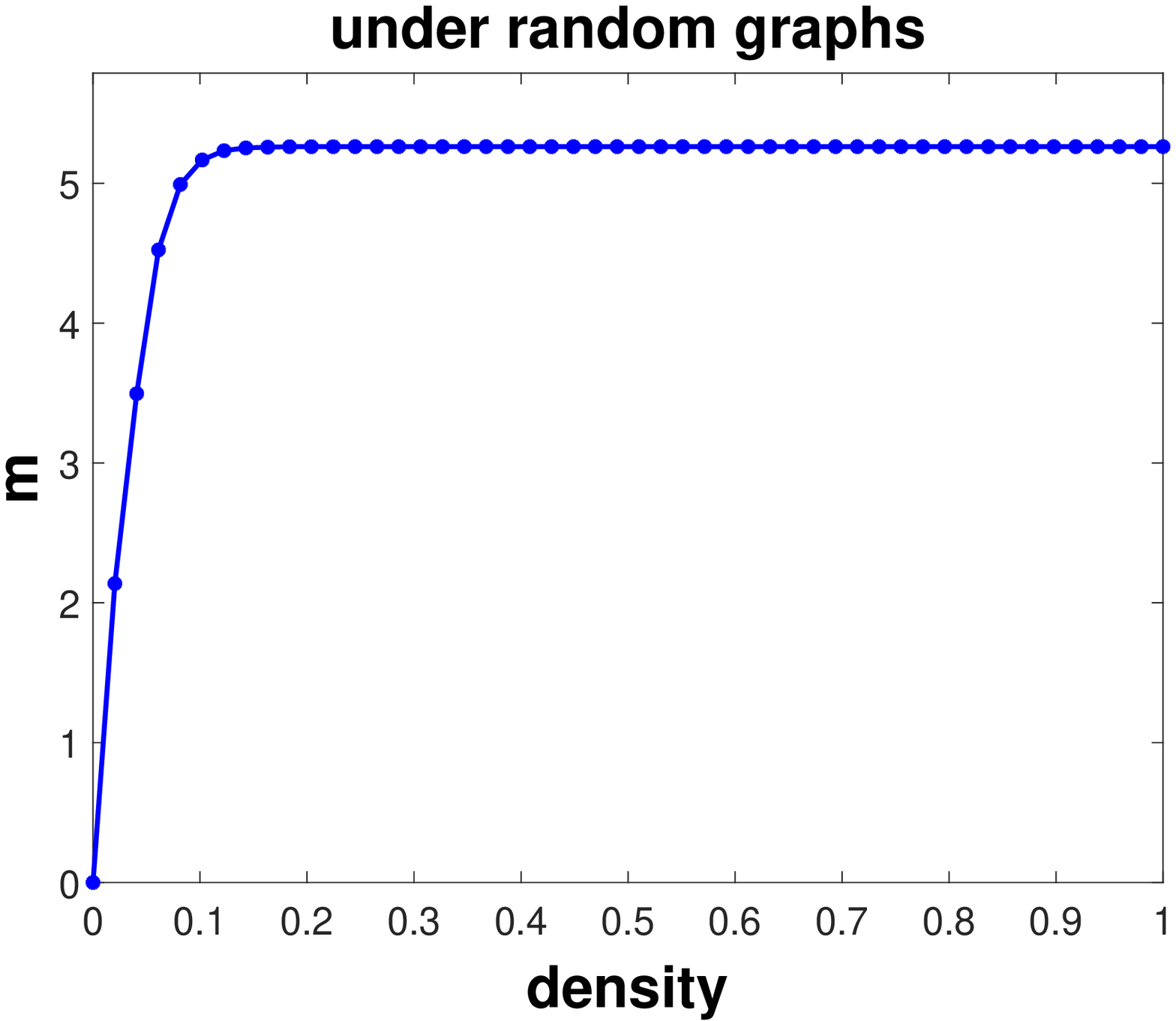}}
\subfloat[] {\includegraphics [width=5cm,  height=4cm]{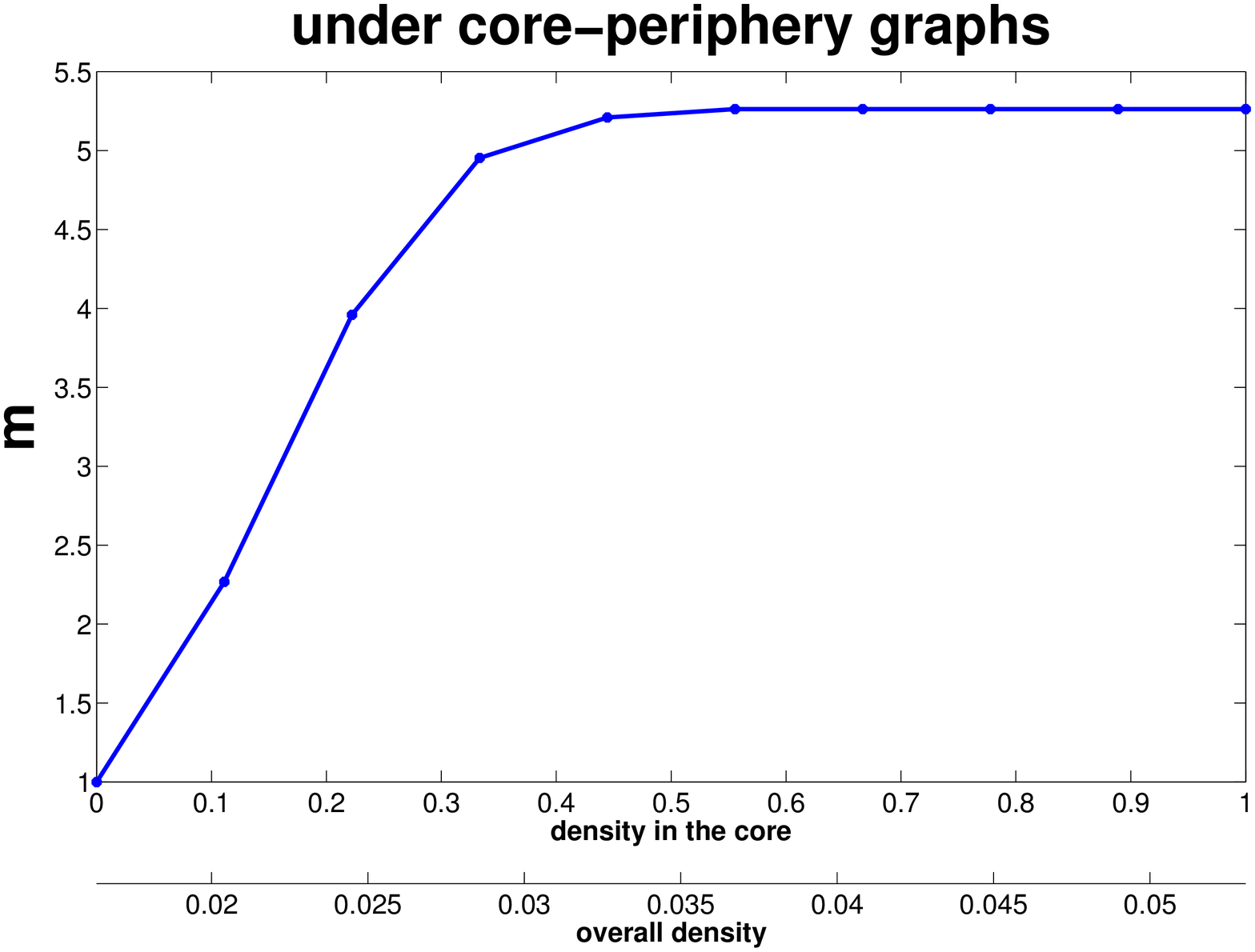}}
\caption{Collateral multiplier ($m$) as a function of density under different network structures. Notice the different scale for the overal density in panel (c).}
\label{fig:collateral_multiplier_example}
\end{figure}

The plots in Figures \ref{fig:collateral_outflowing_example} and
\ref{fig:collateral_multiplier_example} help to visualize the results
contained in the last two propositions. The plots show equilibrium
values of total collateral and of the multiplier resulting from
numerical simulations\footnote{The numerical simulation is implemented with $N=50$,  $h=0.1$,  $1-\theta=0.1$, and $A^0=100$ for all banks.} using each of the three network topologies
examined above (closed k-regular, random graph, and core-periphery)
and with different levels density (of $p_{core}$ for the
core-periphery). First, the plots show that both total collateral and
the multiplier do not change with the level of density in the closed
k-regular graph (plot (a) in both figures). In contrast, both
variables increase with the level of density in the random graph and
in the core-periphery network (respectively, panels (b) and (c) of the
two figures), before eventually converging to the same value of the
close k-regular graph (and determined by the expressions in
Proposition \ref{prop:proposition_positive_out_degree_basic}). The
main intuition for the latter result is that increasing the level of
density (in the overall network or in the core) increases both the
number and the length of cycles in the network\footnote{In addition,
  for the random graph, the number of nodes with positive out-degree
  also increases with density.}. These two factors have a positive
impact on endogenous collateral creation in the network, as we
explained in Section \ref{sec:the role of the direction of collateral
  flows}.  However, when the longest possible cycle in the network
(for the random graph) or in the core (for the core-periphery graph)
adding further links does not longer increase endogenous
collateral. Furthermore, both the third statement of Proposition
\ref{prop:proposition_density_rand} and the plots in the figures
indicate the core-perihery network generates a much higher total
collateral than the random graph already with small increases in
density. For instance the inspection of Figure
\ref{fig:collateral_multiplier_zoom} reveals that already with
$N=50$ banks in the network a tiny increase in
overall density (from $0.02$ to $0.03$) has the effect of more than doubling the value of
the multiplier (from $2$ to almost $5$). In contrast, a much larger change
in density is required to produce a similar effect in the random
graph. This result generalizes the insights discussed in the previous
section (cf. Proposition
\ref{prop:independence-of-cycles-length}). Once all banks have
positive out-degree and they are thus all contributing with outflowing
collateral, concentrating all collateral flows in a small cycle (like
the one in the core) already generates the largest possible total
collateral. This result has also implications for markets
organization, as it indicates that concentrating collateral flows
among few nodes has great advantages for the velocity of collateral
and thus for the overall liquidity of the market. At the same time, in
the next section we shall show that - when liquidity hoarding
externalities are present - networks with highly concentrated
collateral flows are also more exposed to larger collateral hoarding
cascades following small local shocks. 

\begin{figure}[H]
\centering
\includegraphics[scale=0.3]{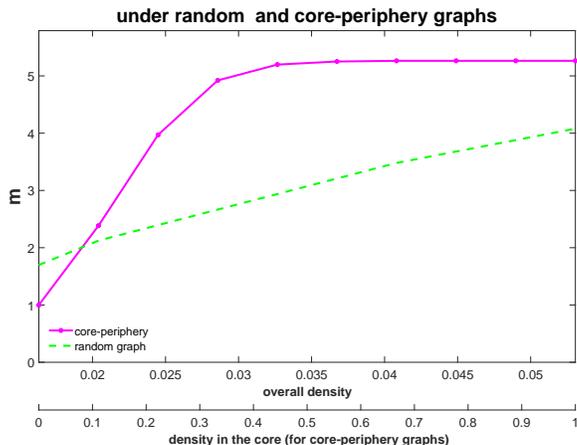}
\caption{Collateral multiplier ($m$) as a function of network the density in the
  random graph and in the core-periphery network. Notice the double scale for the overall network density and the density of links within the core of the core-periphery network.}
\label{fig:collateral_multiplier_zoom}
\end{figure}

\section{Value at Risk and Collateral Hoarding } 
\label{sec: Net liquidity position, Value-at-Risk, and hoarding effects}

So far we have worked with the assumption that non-hoarding rates $\{\theta_i\}_{i=1}^N$ were constant across time and homogeneous across banks. This has simplified the analysis and it has allowed us to highlight the role of the characteristics of network topology in determining collateral flows in the financial system. At the same time, this hypothesis is also quite restrictive as banks' hoarding and non-hoarding  might be responsive to the liquidity risk situation of banks and, accordingly, also by the level of available collateral \citep[see e.g.][]{Acharya2010precautionaryhoarding,Berrospide2012liquidityhoarding,deHaan2013liquidityhoarding}. In addition, 
recent accounts of collateral dynamics \citep{singh2012deleveraging} have documented the sizeable reduction in velocity of collateral in the aftermath of the last financial crisis as a result of increased collateral hoarding by banks. To account for these important phenomena in this section we extend the basic model presented in Section \ref{sec:model} to introduce time-varying non-hoarding rates determined by liquidity risk considerations. 

Again, to keep the model as simple as possible we abstract from many important aspects concerning the liquidity position of the banks. We assume that all funding is secured.   For every  bank $j$ let  $NL_j$ be its net liquidity position.  Recall that the amount of  pledgeable collateral of bank $j$ that can be used to get external funds with a haircut $h$ is $A_j^C$. At the same time, if $L_j  \neq  \varnothing $, a fraction $\theta_j$ of this amount of collateral  is already pledged (i.e. an amount of $A^{C^{out}}= \theta_j A_j^C$). The net liquidity position of the bank $j$  is thus given by:
\begin{equation}
NL_j=(1-h)(1-\theta_j)A_j^C(\theta_1,\theta_2,...,\theta_N, A^0, \mathcal {G})-\epsilon_j,
\label{eq_net_liquidity_position}
\end{equation}
where $\epsilon_j$ are payments due within the periods, i.e. liquidity shocks, which
are assumed to be a i.i.d. normally distributed random variable with mean $\mu_j$ and standard deviation $\sigma_j$. The notation $A_j^C(\theta_1,\theta_2,...,\theta_N,  A^0,  \mathcal {G})$ emphasizes the fact that  the total collateral position of a bank depends also on the fractions of non-hoarded collateral of all banks in the network $ \mathcal {G}$. Notice that the above equation implies that  the more borrowers of $j$ hoard collateral, the lower is the value of collateral $A^C_j$, and thus the higher the need to hoard collateral for $j$.

Let us start by observing that if the liquidity shock is large enough, bank $j$ defaults (i.e. $NL_j < 0$). This occurs when 
\[
\epsilon_j > (1-h)(1-\theta_j) A_j^C.
\]
Given the assumption on the random variable $\epsilon_j$, the default of $j$ is an event occurring with probability
\[ 
\text {prob.} (NL_j<0)=\text {prob.}(\epsilon_j > (1-h)(1-\theta_j)A_j^C).
\]

Furthermore, following  \citet{Adrian2010Liquidity_VaR} and \citet{Adrian2013Leverage_VaR}, we  assume that each bank $j$ employs a  Value-at-Risk (VaR) strategy to determine the fraction $1-\theta_j$ of collateral to hoard, so that the above probability of default is not higher than a target $(1-c_j)$ (where $0<c_j <1$). If we assume that returns on external assets held by $j$ are higher than the repo rate, then each  bank $j$ will decide the optimal fraction $1-\theta_j$ such that
\begin{equation}
\text {prob.}(\epsilon_j > (1-h)(1-\theta_j)A_j^C) = 1-c_j.
\label{eq_VaR_condition}
\end{equation}

Given that  $\epsilon_j$ are a i.i.d. normally distributed random variable  we have
\begin {equation}
\text {prob.} (NL_j<0)=\text {prob.}(\epsilon_j > (1-h)(1-\theta_j)A_j^C)=\frac {1}{2} [1- \text {erf}  (\frac {(1-h)(1-\theta_j)A_j^C-\mu_j}{\sigma_j \sqrt {2}})],
\end {equation}
where  $\text{erf}$ is  Gauss error function  defined as
\begin{equation}
\text {erf} (x)= \frac {1}{\pi}\int_{-x}^{x} \text {e}^{-t^2} dt.
\end{equation}

Under the VaR constraint bank $j$ sets the share of hoarded collateral at the level $\theta^*_j$ such that 
\begin{equation}
\frac {1}{2} [1- \text {erf}  (\frac {(1-h)(1-\theta_j)A_j^C-\mu_j}{\sigma_j \sqrt {2}})]= 1-c_j
\label{eq_VaR_condition_normal1}
\end{equation}
$ \Leftrightarrow $
\begin{equation}
 \text {erf}  (\frac {(1-h)(1-\theta_j)A_j^C-\mu_j}{\sigma_j \sqrt {2}})= 2c_j-1
\label{eq_VaR_condition_normal2}
\end{equation}
$ \Leftrightarrow $ 
\begin{equation}
\theta_j= 1- \frac {\sigma_j \sqrt {2}  \text {argerf} (2c_j-1) +\mu_j}{(1-h) A_j^C},
\label{eq_VaR_condition_normal3}
\end{equation}
where  $\text {argerf}$ is the inverse error function  defined  in $(-1, 1) \to \mathbb{R}$ such that  
\begin{equation}
\text {erf}(\text {argerf}(x))	=x.
\end{equation}

Equation (\ref {eq_VaR_condition_normal3}) indicates that $\theta_j$ is a decreasing function of the VaR target $c_j$, of the uncertainty about the liquidity shock (captured by $\sigma_j$), of the mean of  the liquidity shock $\mu_j$, and of the haircut rate $h$. Moreover, it is an increasing function of value of the collateral $A^C$ as well as of the shares of non-hoarded collateral of other banks in the network. 
Denote 
\begin{equation}
c_j^0=\sigma_j \sqrt {2}  \text {argerf} (2c_j-1) +\mu_j,
\label{eq_c_0_Norm}
\end{equation}
 we then obtain  the following final expression for the optimal $\theta_j$ under  the  assumption of normally-distributed liquidity shocks.
\begin{equation}
\theta_j= 1- \frac {c_j^0}{(1-h) A_j^C(\theta_1,\theta_2,...,\theta_N, A^0, \mathcal {G})}.
\label{eq_VaR_theta_Norm}
\end{equation}

Notice that the endogenous level of non-hoarding, $\theta_j$, depends now not only on the uncertainty about $\epsilon_j$ but also on the haircut rate $h$, as well as on the value of the collateral $A^C_j$. The interdependence between $\theta_j$ and $A^C_j$ implies that each bank will adjust its hoarding preference (i.e. hold more or less collateral) in anticipation of expected ``losses'' or ``gains'' in  its total amount of collateral. It also follows that a change in hoarding rates at bank $j$ will induce a change in hoarding rates at banks to which $j$ is connected to. 

We now investigate the existence of equilibria in non-hoarding rates. Let us start by noticing that Equation  (\ref{eq_VaR_theta_Norm}) indicates that, for every bank j, if $L_j  \neq  \varnothing $, its hoarding can in general be expressed as:  
\begin {equation}
 (1- \theta_j)=  \frac {c_j^0}{(1-h) A_j^C(\theta_1,\theta_2,...,\theta_N, A^0, \mathcal {G})},
\label{eq_theta_VaR}
\end {equation}
where the variable $c_j^0$ defined by equation (\ref {eq_c_0_Norm})  captures the effects of  uncertainty in $\epsilon_j$ on  $(1- \theta_j)$. Since $\theta_j \in [0, 1]$, it follows that $A_j^C(\theta_1,\theta_2,...,\theta_N, A^0, \mathcal {G})$ must be  in $[\frac {c_j^0}{1-h}, \infty]$.
From equation (\ref{eq_theta_VaR}), it follows that
\begin{equation}
(1-h) \theta_j=\frac {(1-h)A^C_j- c^0_j}{ A^C_j}.
\label{eq_theta_VaR2}
\end{equation}
Equivalently,
\begin{equation}
(1-h) \frac {\theta_j}{k_j^{out}} =\frac {(1-h)A^C_j- c^0_j}{A^C_j k_j^{out}}.
\label{eq_VaR_share_dynamic}
\end{equation}

Recall that under the assumption that banks homogenously spread collateral across their lenders, we have
\[
w_{i \leftarrow j}=\theta_j s_{i\leftarrow j}= \frac{\theta_j}{k_j^{out}}, \ \forall i \in L_j \neq  \varnothing.
\]   
Therefore,
\[
(1-h) w_{i \leftarrow j}= \frac {(1-h)A^C_j - c^0_j}{A^C_j k_j^{out}}, \ \forall i \in L_j \neq  \varnothing.
\]   
Next, denote by $\mathcal {\tilde {W}}^{VaR}= \{ \tilde {w}_{i\leftarrow j} \}$ the matrix with size ($N\text{x}N$) where
\begin {equation}
\begin {cases}
\tilde {w}_{i\leftarrow j}=\frac {(1-h)A^C_j - c^0_j}{A^C_j k_j^{out}}, \ \forall i \in L_j \neq  \varnothing\\
\tilde {w}_{i\leftarrow j}=0,  \  \mbox {elsewhere}\\
\end {cases}
\label{w_VaR_model}
\end {equation}
The level of collateral $A_i^C$ is then obtained by solving the following equation\begin{equation}
A_i^C=A^0_i+\sum_{j \in B_i} \tilde {w}_{i \leftarrow j} A^C_j \ \forall i=1,2,...N.
\label{eq_collateral_box_VaR_mt}
\end{equation}
Finally, the non-hoarded rates $\theta_i$ can be obtained by substituting  $A_i^C$ into equation  (\ref{eq_theta_VaR2}).

In general, the solution to the system composed by the system of
equations in (\ref{eq_collateral_box_VaR_mt}) might not be unique. However, the following proposition establishes sufficient conditions for the uniqueness of the solution.
\begin{proposition}
Let $\tilde {b}$ be a column vector size $N \text{x}1$, given by   
\begin {equation}
\tilde {b}= A^0- \mathcal {S}\mathcal {C}^0.
\label {parameter_b_VaR}
\end {equation}
Define the matrix $\tilde {\mathcal{A}}=\{ \tilde {a}_{i\leftarrow j}\}$ with size $N$\text {x}$N$ as 
\begin {equation}
\tilde {\mathcal{A}} =\mathcal {I}-(1-h)\mathcal {S}.
\label{parameter_A_VaR}
\end {equation}
If $0 < h < 1$ and $\tilde {\mathcal{A}}^{-1} \tilde {b} \geq \frac {\mathcal {C}^0}{1-h}$ where  $\mathcal {C}^0=[c^0_1, c^0_2,..., c^0_{N-1}, c^0_N]^T$ is the column vector capturing the effects of the net liquidity shock on hoarding preferences, then the system  (\ref{eq_collateral_box_VaR_mt}) has the unique solution\footnote {Throughout this paper, for any  two  vectors X and Y of size n, then  $X \geq Y$ if $X_i \geq Y_i, \forall i=1,...n.$}:
\begin {equation}
A^{C}=\tilde {\mathcal{A}}^{-1}\tilde {b}.
\label {eq_total_collateral_solution}
\end {equation}
\label{proposition_solution_A_C_VaR}
\end{proposition}
\begin{proof} 
See the appendix.
\end{proof}

Finally, by substituting $A^{C}$  in (\ref{eq_total_collateral_solution}) into equation  (\ref{eq_theta_VaR2}), we obtain the solution to the equilibrium rates of non-hoarded collateral as follows
\begin {equation}
\theta_j=  1-  \frac {c_j^0}{(1-h) \tilde {\mathcal{A}}^{-1}\tilde {b}}.
\label{eq_theta_VaR_solution}
\end {equation}
In the next section, we use the results obtained from the above proposition about the determination of equilibrium banks' collateral $A_j^{C}$ and non-hoarding rates $\theta_j$ to study the emergence of collateral hoarding cascades under different network structures when some banks are hit by uncertainty shocks, captured by an increase in the variable $c^0_j$.

\section{Collateral hoarding cascades}
\label{sec:coll-hoard-casc}

We now use the VaR collateral hoarding model developed in the previous
section to study how different structures of rehypothecation networks
react when a fraction of banks in the network is hit by adverse
shocks.  In particular we focus on uncertainty shocks\footnote{We also
  performed analysis of the impact of aggregate shocks to the value of
  collateral. However, results and the dynamics were similar to the
  one reported in the paper.} that cause the variable $c^0_i$ in
Equation \eqref{eq_VaR_theta_Norm} to increase to $c^1_i=
c^0_i(1+\tilde {c}^0)$ for some banks $i$, where $\tilde {c}^0 \geq
0$. The rise in uncertainty will lead those banks to increase their
hoarding of collateral. In turn, this will trigger a cascade of
hoarding effects at banks that are either directly or indirectly
connected to the banks initially it by the uncertainty shock. Indeed,
higher hoarding at bank $i$ will also cause a loss in the collateral
flowing into bank's $i$ neighbors. As a consequence, the latter banks
will also increase their hoarding rates, causing further loss in
collateral inflows at other banks in the system and thus further
adjustments in hoarding rates at other banks in the system. The final
result of the foregoing cascade will be a new equilibrium
characterized, in general, by a lower level of total collateral in the
system. On the grounds of the results stated in Proposition
\ref{proposition_solution_A_C_VaR} we can determine the new
equilibrium vector of collateral in the aftermath of the local
uncertainty shock. More formally, let $\mathcal {C}^1=[c^1_1, c^1_2,..., c^1_{N-1}, c^1_N]^T$ the new vector of uncertainty factors in the aftermath of the local uncertainty shock. By substituting $\mathcal {C}^1$ into  equations (\ref{eq_total_collateral_solution}) and  (\ref{parameter_b_VaR}),  we obtain the new equilibrium solution, $A^{C^{1}}$, to  $A^{C}$ at the end of the hoarding  cascade as
\begin {equation}
A^{C^{1}}=\tilde {\mathcal{A}}^{-1} [A^0 - \mathcal {S}\mathcal {C}^1].
\label {eq_total_collateral_solution_shock_uncertainty}
\end {equation}

\begin{figure}[H]
\centering
\captionsetup[subfloat]{farskip=0pt,captionskip=0pt}
\subfloat[random attack] {\includegraphics [width=8cm,  height=6cm]{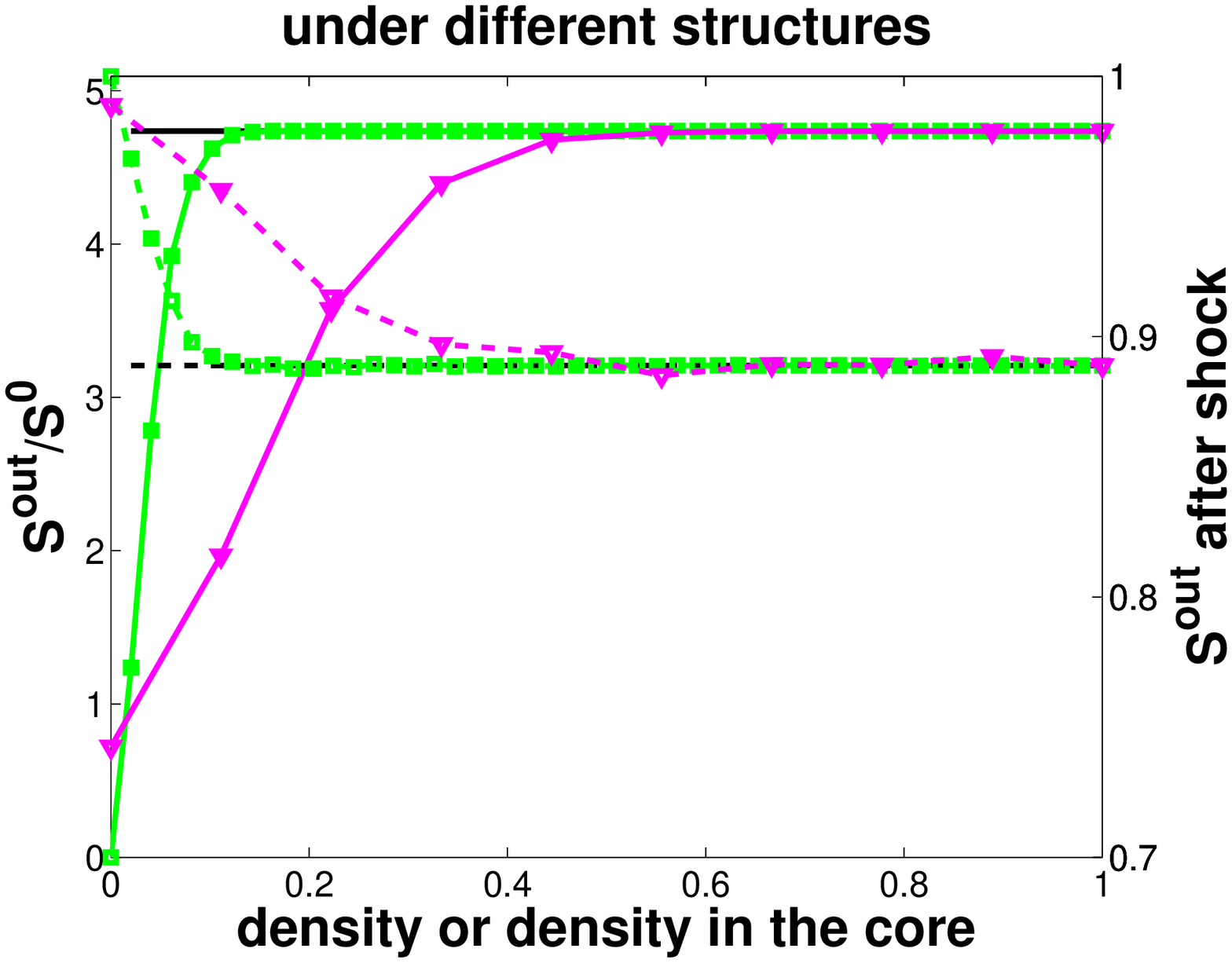}}
\subfloat [targeted attack] {\includegraphics [width=8cm,  height=6cm ]{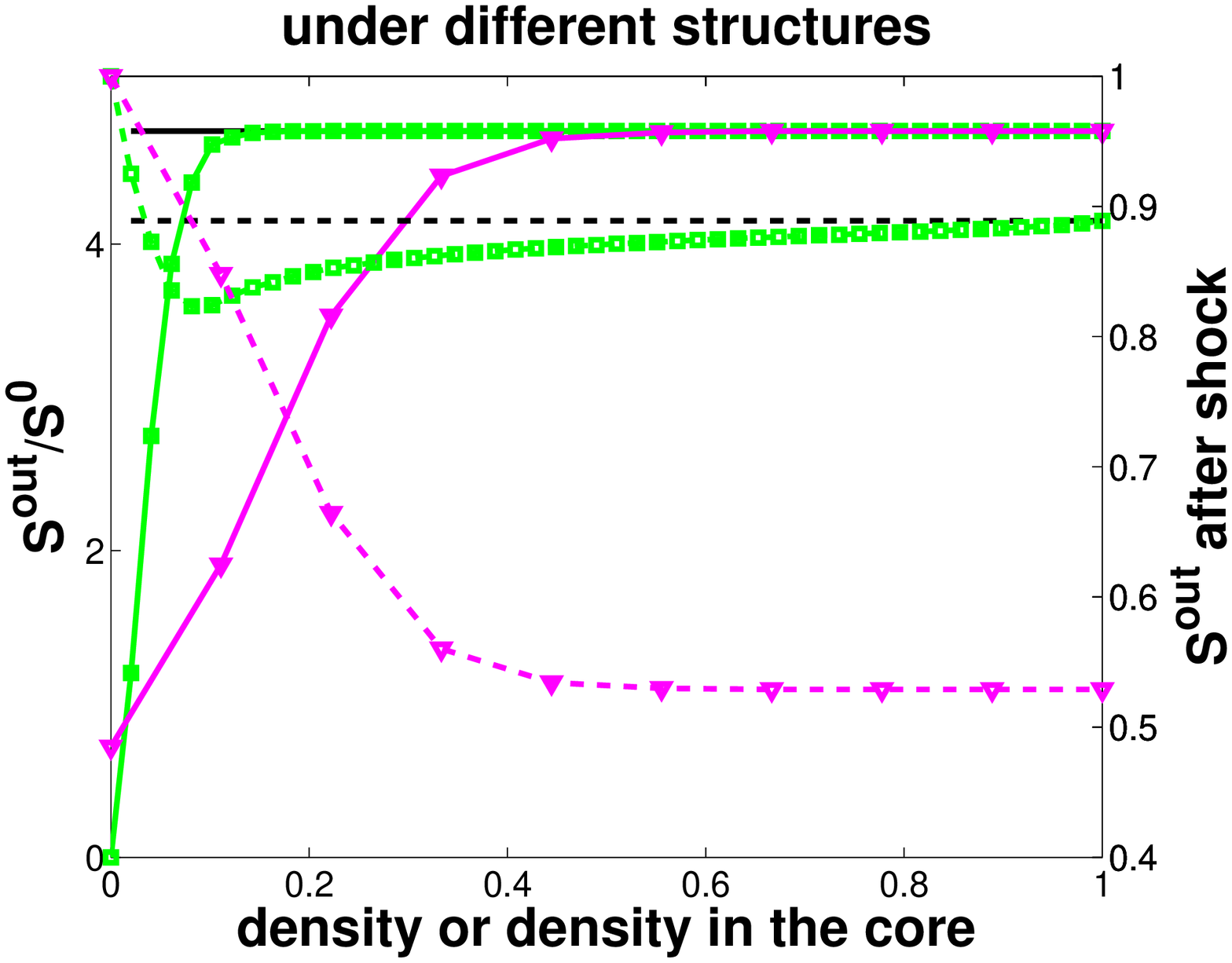}}
\caption{Collateral multiplier and collateral losses following local uncertainty shocks, in different network structures and for different overall density levels.  In each panel,  the left y-axis  shows the ratio between total outflowing collateral and total proprietary collateral $S^{out}/S^{0}$ (solid line). The right y-axis shows total collateral $S^{out}$  after shock (relative to the pre-shock level, dashed line). Different network structures are represented by different colors and markers: closed k-regular graph (black, no marker), random graphs (green, squares), core-periphery graphs (magenta, diamonds).}
\label{compare_structures_extended_model_hoarding}
\end{figure}

Let us now investigate how after-shock equilibrium collateral  of
behaves in different networks.  The plots in Figure
\ref{compare_structures_extended_model_hoarding} show the pre-shock
levels of the ratio between equilibrium total outflowing collateral
$S^{out}$ and total proprietary collateral $S^{0}$ (left
scale). Moreover, they show 
equilibrium total ouflowing collateral after the shock relative to its
pre-shock value (right-scale). The two variables are plotted as functions
of network density and for the three different structures examined in
Section \ref{sec:netw-arch-coll}, namely the closed k-regular
network $\mathcal{G}_{reg}$, the random network $\mathcal{G}_{rg}$,
and the core-periphery network $\mathcal{G}_{cp}$. The plots refers to
numerical investigations of equilibrium collateral before and after a fraction
$f=20\%$ of banks in the network experiences a $50\%$ increase in the
uncertainty factor (so that $\tilde {c}^0=0.5$) and are performed
under two different shocks scenarios. In the first of them (``random
attack'', left-side plots) the banks hit by the uncertainty shock are
randomly selected. In the second scenario (``targeted attack'',
right-side plots) the shocked banks are selected according to their
centrality\footnote{We use ``PageRank'' centrality \citep[see][for
  more details]{newman2010networks,battiston2012debtrank}. However,
  the main conclusions still hold when the degree centrality is
  employed.} in the network (in descending order).\footnote{The other
parameters of the simulation where set so that the number of banks is
$N=50$, and $h=0.1$, $A_i^0=100$, $1-\theta_i=0.1$  and $c^0_i=1$ for all
banks.} Notice that the second scenario is relevant only for the random network and for the core-periphery, as all nodes have the same centrality in the closed k-regular network. 

The analysis of Figure
\ref{compare_structures_extended_model_hoarding} reveals first that -
before the shock - the behaviour of total collateral under the three network structures is the same as the one discussed in Section \ref{sec:netw-arch-coll} for the case of constant hoarding rates\footnote{Indeed, all the results of Propositions \ref{prop:new_prop_k_reg} and \ref{prop:proposition_density_rand} hold also in the model where hoarding rates are determined according to a VaR criterion. For the sake of brevity, we do not report here these propositions and the related proofs. However, they are available from the authors upon request.}. Second, the effects of the local uncertainty shock vary with density only in the random network and in the core-periphery network. In contrast, the loss in total collateral does not change with density in the closed k-regular graph. Third, the overall impact of the shock is very different across the scenarios considered. In particular, the maximal impact is quite small (only a $10\%$ loss) for all the three structures in the random attack scenario. In addition, all the structure generate the same total loss as overall density converges to 1. In contrast, in the targeted attack scenario, the total loss generated in presence of a core-periphery network is much larger than in the other network structures (up to $50\%$ of the initial total collateral value) and increases with density. 

Thus, the core-periphery network generates very large losses in
collateral when central nodes are hit by local uncertainty
shocks. Recall that those nodes are precisely the one concentrating
collateral flows, and thus generating the large increases in
endogenous collateral stressed in Section \ref{sec:netw-arch-coll}. It
follows that the core-periphery network displays a trade-off between
liquidity and systemic (liquidity) risk. On the one hand,
concentrating collateral flows generates large gains in market
liquidity already with a low density of the network. On the other
hand, high concentration produces large liquidity losses when small
local shock come across. Indeed, by concentrating collateral flows,
more central nodes have also a larger impact on collateral at the
peripheral nodes that have connections  with them, and they thus
trigger a must stronger adjustment in hoarding rates thereafter. Notice that these concentration effects underlying the large liquidity losses in the core-periphery network are much smaller in the random graph (where link heterogeneity is small). In addition, they are completely absent in the closed k-regular graph, where all nodes have the same centrality. 

To conclude, it is useful to stress that the presence of liquidity hoarding externalities is central for the above results about the trade-off observed in the core-periphery networks. In other terms, targeted local shocks would have a small impact in core-periphery networks if hoarding rates were not responsive to changes in the liquidity positions of banks. To understand why, notice that with constant hoarding rates a loss in collateral value or an exogenous increase in hoarding rates at one bank $i$ will only have a $n^{th}$-order effect at other banks directly or indirectly connected to it. And this is because the initial shock is dampened by haircut rates and hoarding rates at other banks along the chain. For instance, in the very simple case of the cycle (i.e. closed chain) displayed in Figure \ref{fig_three_nodes} (c) a shock to outflowing collateral at node 1  ($A^{0^{out}}_1$) will have an effect only of order $\theta_3  (1-h) \theta_2  (1-h)$ on outflowing collateral at node 3 (see also Equation \eqref{eq:closed_chain_system}). In contrast, in the case of VaR-determined hoarding rates, the effect can be large because it is reinforced by the process of agents' revising their hoarding rates at each step along the rehypothecation chain.\footnote{In particular, in the example of Figure \ref{fig_three_nodes} (c) with VaR-based hoarding rates the effect of a shock to $A^{0^{out}}_1$ on the outflowing collateral of node 3 be of the order $\left(1-h\right)^2\left(\theta_3  \theta_2+\theta_3\frac{\partial \theta_2}{\partial A^{0^{out}}_1}+\theta_2\frac{\partial \theta_3}{\partial A^{0^{out}}_1}\right).$}

\section{Concluding remarks}
\label{Conclusions}

We have introduced and analyzed a simple model to study collateral flows over a network of repo contracts among banks. We have assumed that, to obtain secured funding, banks may pledge their proprietary collateral or re-pledge the collateral obtained by other banks via reverse repos. The latter practice is known as ``rehypothecation'' and it has clear advantages for market liquidity as it allows banks to secure more transactions with the same set of collateral. At the same time, re-pledging other banks' collateral may also raise liquidity risk concerns, as several banks rely on the same collateral for their repo transactions. We have focused on investigating which characteristics of rehypothecation networks are key in order to increase the velocity of collateral flows, and thus increase the liquidity in the market. We have first assumed that banks hoard a constant fraction of their collateral. Under this hypothesis we have shown that characteristics of the network like the length of re-pledging chains, the presence of cyclic chains or the direction of collateral flows are key determinant of the level of endogenous collateral in the system, defined as a overall level of collateral larger than the initial proprietary endowments of banks. In particular, we have shown that the level of endogenous collateral increases with chains' length. However, cyclic chains allow, \textit{ceteris paribus}, for a larger endogenous collateral than a-cyclic chains of the same length. Finally, we have shown that endogenous collateral is large already with small cyclic chains if the banks involved centralize collateral flows. The foregoing features of network topology underlie the results about the determination of total collateral in more general network architectures. In particular, we showed that total collateral increases with density in the random network (displaying a mild heterogeneity in collateral flows across banks) and in the core-periphery networks (displaying high concentration of collateral flows). Nevertheless, core-periphery networks generate larger collateral than random networks already with smaller increase in the density of the network. The foregoing results have implications for the micro-structure of markets where collateral is important, as they highlight a new factor besides network density - i.e. concentration in collateral flows - that allows for significant gains in market liquidity. A market with highly concentrated collateral flows generates higher velocity of collateral, as it is thus more liquid, even if banks are not tied by a dense network of financial contracts. 

Furthermore, we have extended the model to allow for endogenous levels of hoarding rates that depend on the liquidity position of each bank in the network. More precisely, we assumed that banks set their hoarding rates by adopting a Value-at-Risk criterion aimed at minimizing the risk of liquidity defaults. We have shown that, in these framework, hoarding rates of each single bank are in general dependent on other banks' rates and collateral levels, a feature which introduces important collateral hoarding externalities in the analysis. We have then used the above framework to study the overall impact on collateral flows of local adverse shocks leading to an increase in payments' uncertainty at some banks in the market, in particular investigating how the response may vary with the topology of the rehypothecation network. We have highlighted that core-periphery networks generate larger losses in overall collateral compared to other network structures (closed k-regular network, random network) when the nodes experiencing the shock are the most central nodes in the network, i.e. the ones concentrating collateral flows. This result has interesting implications for the regulatory analysis of markets with collateral, as it shows that the same network structures allowing for the largest increase in collateral velocity - i.e. core-periphery networks - are also the ones more exposed to largest collateral hoarding cascades in case of local shocks. A trade-off between liquidity and systemic (liquidity) risk thus emerges in those networks. In addition, our results also suggest that - in the presence of rehypothecation - regulatory liquidity and collateral requirements imposed to banks need to account for the structure of the network of lending contracts across banks. In particular, such requirements should both account for the systemic role (e.g. centrality) of the banks in the collateral flow network, as well as account for the whole topology of the network (e.g. for the presence of hierarchical structures such as the core-periphery architecture).

Our work could be extended at least in three ways. First, in these work we have abstracted from many important aspects of real-world secured lending markets, such as heterogeneous collateral quality and endogenous haircut rates. Introducing these elements could probably enrich our results. In particular, in the model the haircut rate affects the extent of rehypothecation. Accordingly, endogenous haircut rates that reflect different levels of counterparty risk can be an additional source of externalities in the model. Second, we have used the Value at Risk criterion to determine hoarding rates in our model. However, hoarding rates may also result from liquidity requirements imposed to banks. It would then be interesting to extend the model to study how different requirements that have been proposed so far may impact on market liquidity in presence of rehypothecation, and how these requirements should be designed in order to minimize the liquidity-systemic risk trade-off that we highlighted above. Finally, we have focused on bilateral repo contracts. However, it would be interesting to study how our results might change in presence of try-party repo structures, and in the presence of central clearing counterparties that interact with banks re-using their collateral. 


\singlespacing
 
\section*{Acknowledgments}

\footnotesize{
We are indebted to Joseph E. Stiglitz, Stephen G. Cecchetti, S\'{e}rafin Jaramillo, Dilyara Salakhova, Marco D'Errico, Guido Caldarelli, for valuable comments and discussions that helped to improve the paper. We also thank participants to the various conferences where earlier versions of this paper were presented. These include the Second Conference on Network Models and Stress Testing for Financial Stability, Banco de M\'{e}xico, September 26-27, 2017, the Conference on Complex System 2017 (CCS 2017), in Canc\'{u}n, September 17-22, 2017, the first and second FINEXUS Conference, Z\"urich (January 2017, January 2018), the 22nd Workshop on Economic Science with Heterogeneous Interacting Agents (WEHIA 2017), Catholic University of Milan, June 12-14, 2017, and the 23rd Computing in Economics and Finance (CEF 2017), Fordham University, New York City, June 28-30, 2017. We also participants to seminars at the Jaume I University, Castell\`{o}n de la Plana (Spain) in November 2017 and to the Catholic University of Milan, in December 2017. All usual disclaimers apply. The authors gratefully acknowledge the financial support of the Horizon 2020 Framework Program of the European Union under the grant agreement No. 640772 - Project DOLFINS (Distributed Global Financial Systems for Society). }


\bibliographystyle{apalike}
\bibliography{references}

\begin{thebibliography}{}

\bibitem[Acharya and Merrouche, 2010]{Acharya2010precautionaryhoarding}
Acharya, V. and Merrouche, O. (2010).
\newblock Precautionary hoarding of liquidity and inter-bank markets: Evidence
  from the sub-prime crisis.
\newblock {\em Working Paper 16395}.

\bibitem[Adrian and Shin, 2010]{Adrian2010Liquidity_VaR}
Adrian, T. and Shin, H.~S. (2010).
\newblock Liquidity and leverage.
\newblock {\em Journal of Financial Intermediation}, 19(3):418--437.

\bibitem[Adrian and Shin, 2014]{Adrian2013Leverage_VaR}
Adrian, T. and Shin, H.~S. (2014).
\newblock Procyclical leverage and value-at-risk.
\newblock {\em Review of Financial Studies}, 27(2):373-- 403.

\bibitem[Aguiar et~al., 2016]{Aguiar_Map_Collateral2016}
Aguiar, A., Bookstaber, R., Kenett, D.~Y., and Wipf, T. (2016).
\newblock A map of collateral uses and flows.
\newblock {\em Journal of Financial Market Infrastructures}, 5(2):1--28.

\bibitem[Andolfatto et~al., 2017]{andolfatto2017rehypothecation}
Andolfatto, D., Martin, F.~M., and Zhang, S. (2017).
\newblock Rehypothecation and liquidity.
\newblock {\em European Economic Review}, 100:488--505.

\bibitem[Battiston et~al., 2016]{battiston2016complexity}
Battiston, S., Farmer, J.~D., Flache, A., Garlaschelli, D., Haldane, A.~G.,
  Heesterbeek, H., Hommes, C., Jaeger, C., May, R., and Scheffer, M. (2016).
\newblock Complexity theory and financial regulation.
\newblock {\em Science}, 351(6275):818--819.

\bibitem[Battiston et~al., 2012]{battiston2012debtrank}
Battiston, S., Puliga, M., Kaushik, R., Tasca, P., and Caldarelli, G. (2012).
\newblock Debtrank: Too central to fail? financial networks, the fed and
  systemic risk.
\newblock {\em Scientific reports}, 2:541.

\bibitem[Berrospide, 2012]{Berrospide2012liquidityhoarding}
Berrospide, J. (2012).
\newblock Precautionary hoarding of liquidity and inter-bank markets: Evidence
  from the sub-prime crisis.
\newblock {\em Working Paper}.

\bibitem[Bottazzi et~al., 2012]{bottazzi2012securities}
Bottazzi, J.-M., Luque, J., and P{\'a}scoa, M.~R. (2012).
\newblock Securities market theory: Possession, repo and rehypothecation.
\newblock {\em Journal of Economic Theory}, 147(2):477--500.

\bibitem[Brunnermeier and Pedersen, 2008]{brunnermeier2008market}
Brunnermeier, M.~K. and Pedersen, L.~H. (2008).
\newblock Market liquidity and funding liquidity.
\newblock {\em The Review of Financial Studies}, 22(6):2201--2238.

\bibitem[Capel and Levels, 2014]{Capel2014collateral}
Capel, J. and Levels, A. (2014).
\newblock Collateral optimisation, re-use and transformation: Developments in
  the dutch financial sector.
\newblock {\em DNB Occasional Studies}, 12(5):1–--54.

\bibitem[de~Haan and van~den End, 2013]{deHaan2013liquidityhoarding}
de~Haan, L. and van~den End, J.~W. (2013).
\newblock Banks’ responses to funding liquidity shocks: Lending adjustment,
  liquidity hoarding and fire sales.
\newblock {\em Journal of International Financial Markets, Institutions and
  Money}, 26:152--174.

\bibitem[{Financial Stability Board}, 2017a]{FSB2017_measure}
{Financial Stability Board} (2017a).
\newblock Non-cash collateral re-use: Measure and metrics.
\newblock 25 January:1--17.

\bibitem[{Financial Stability Board}, 2017b]{FSB2017}
{Financial Stability Board} (2017b).
\newblock Re-hypothecation and collateral re-use: Potential financial stability
  issues, market evolution and regulatory approaches.
\newblock 25 January:1--42.

\bibitem[Gai et~al., 2011]{gai2011complexity}
Gai, P., Haldane, A., and Kapadia, S. (2011).
\newblock Complexity, concentration and contagion.
\newblock {\em Journal of Monetary Economics}, 58:453--470.

\bibitem[Gorton and Metrick, 2012]{gorton2012securitized}
Gorton, G. and Metrick, A. (2012).
\newblock Securitized banking and the run on repo.
\newblock {\em Journal of Financial economics}, 104(3):425--451.

\bibitem[Gottardi et~al., 2017]{GottardiMaurinMonnet2017}
Gottardi, P., Maurin, V., and Monnet, C. (2017).
\newblock A theory of repurchase agreements, collateral re-use, and repo
  intermediation.
\newblock CESifo Working Paper Series 6579, CESifo Group Munich.

\bibitem[Leitner, 2011]{Leitner2011}
Leitner, Y. (2011).
\newblock Why do markets freeze?
\newblock {\em Business Review, Federal Reserve Bank of Philadelphia},
  (Q2):12--19.

\bibitem[Monnet, 2011]{monnet2011rehypothecation}
Monnet, C. (2011).
\newblock Rehypothecation.
\newblock {\em Business Review, Federal Reserve Bank of Philadelphia},
  (Q4):18--25.

\bibitem[Newman, 2010]{newman2010networks}
Newman, M. (2010).
\newblock {\em Networks: an introduction}.
\newblock Oxford university press.

\bibitem[Pozsar and Singh, 2011]{pozsar2011nonbank}
Pozsar, Z. and Singh, M. (2011).
\newblock The nonbank-bank nexus and the shadow banking system.
\newblock Working Paper 11-289, International Monetary Fund.

\bibitem[Singh, 2011]{singh2011velocity}
Singh, M. (2011).
\newblock Velocity of pledged collateral: analysis and implications.
\newblock {\em IMF Working Paper 11/256}.

\bibitem[Singh, 2012]{singh2012deleveraging}
Singh, M. (2012).
\newblock The (other) deleveraging.
\newblock {\em IMF Working Paper 12/179}.

\bibitem[Singh, 2016]{singh2016collateral}
Singh, M. (2016).
\newblock {\em Collateral and financial plumbing}.
\newblock Risk Books.

\end{thebibliography}

\clearpage

\section{Appendix: Proofs of Propositions}
\label{Appendix}

\singlespacing

\footnotesize{
\textbf {Proof of proposition \ref{prop:independence-of-cycles-length}}
  \medskip

\noindent Let us start by remarking that, in all cases shown in Figure (\ref{fig_five_nodes}), we have  $k^{out}_i >0$ ($\forall i=1, 2, 3, 4, 5$), and thus  $\delta_i =1$  ($\forall i=1, 2, 3, 4, 5$) and 
\begin{equation}
\begin{cases}
A^{C^{out}}_{1,t=1}=A^{0^{out}}_1  = \theta_1 A^{0}_{1} \\ 
A^{C^{out}}_{2,t=1} =A^{0^{out}}_2 =  \theta_2 A^{0}_{2}   \\ 
A^{C^{out}}_{3,t=1} =A^{0^{out}}_3  =\theta_3 A^{0}_{3} \\ 
A^{C^{out}}_{4,t=1} =A^{0^{out}}_4 =  \theta_4 A^{0}_{4}  \\ 
A^{C^{out}}_{5,t=1} =A^{0^{out}}_5  =\theta_5 A^{0}_{5}\\ 
\end{cases}
\label {eq_five_nodes_initial}
\end{equation}

In the case of  Figure \ref{fig_five_nodes} (a),  the dynamics of $\{A^{C^{out}}_{i,T}\}_{i=1}^5$ will follow
\begin{equation}
\begin{cases}
A^{C^{out}}_{1,T+1}= A^{0^{out}}_1\\
A^{C^{out}}_{2,T+1}= A^{0^{out}}_2+ (1-h)\theta_2 A^{C^{out}}_{1,T}\\
A^{C^{out}}_{3,T+1}= A^{0^{out}}_3+(1-h)\theta_3 A^{C^{out}}_{2,T} + (1-h) \theta_3   A^{C^{out}}_{5,T}\\ 
A^{C^{out}}_{4,T+1}=  A^{0^{out}}_4  + (1-h) \theta_4 A^{C^{out}}_{3,T}\\ 
A^{C^{out}}_{5,T+1}=  A^{0^{out}}_5  + (1-h) \theta_5  A^{C^{out}}_{4,T}\\ 
\end{cases}
\end{equation}
We obtain 
\begin{equation}
\colvecfive[A^{C^{out}}_{1,T+1}]{A^{C^{out}}_{2,T+1}}{A^{C^{out}}_{3,T+1}} {A^{C^{out}}_{4,T+1}}{A^{C^{out}}_{5,T+1}}= \{ I+ [(1-h)M_a]^1+ [(1-h) M_a]^2+...+ [(1-h)M_a]^T \}
\colvecfive[A^{0^{out}}_1]{ A^{0^{out}}_2}{A^{0^{out}}_3} {A^{0^{out}}_4}{A^{0^{out}}_5},
\label{eq_five_nodes_1}
\end{equation}
with 
 \[M_a=\begin{bmatrix}
    0    &  0    &  0    &  0   & 0 \\
\theta_2    &  0    &  0    &  0   & 0 \\
    0    &  \theta_3     &  0    &  0   &  \theta_3   \\
0      & 0      &     \theta_4      & 0   & 0 \\
0      &   0      &  0      & \theta_5& 0
\end{bmatrix}.\]

In the case of  Figure \ref{fig_five_nodes} (b), the length of the closed cycle is 4, and now the dynamics of $\{A^{C^{out}}_{i,T}\}_{i=1}^5$ is governed by the following system

\begin{equation}
\begin{cases}
A^{C^{out}}_{1,T+1}= A^{0^{out}}_1\\
A^{C^{out}}_{2,T+1}= A^{0^{out}}_2+ (1-h)\theta_2 A^{C^{out}}_{1,T} + (1-h) \theta_2   A^{C^{out}}_{5,T}\\
A^{C^{out}}_{3,T+1}= A^{0^{out}}_3+(1-h)\theta_3 A^{C^{out}}_{2,T} \\ 
A^{C^{out}}_{4,T+1}=  A^{0^{out}}_4  + (1-h) \theta_4 A^{C^{out}}_{3,T}\\ 
A^{C^{out}}_{5,T+1}=  A^{0^{out}}_5  + (1-h) \theta_5  A^{C^{out}}_{4,T}\\ 
\end{cases}
\end{equation}

We obtain 
\begin{equation}
\colvecfive[A^{C^{out}}_{1,T+1}]{A^{C^{out}}_{2,T+1}}{A^{C^{out}}_{3,T+1}} {A^{C^{out}}_{4,T+1}}{A^{C^{out}}_{5,T+1}}= \{ I+ [(1-h)M_b]^1+ [(1-h) M_b]^2+...+ [(1-h)M_b]^T \}
\colvecfive[A^{0^{out}}_1]{ A^{0^{out}}_2}{A^{0^{out}}_3} {A^{0^{out}}_4}{A^{0^{out}}_5},
\label{eq_five_nodes_2}
\end{equation}
with 
 \[M_b=\begin{bmatrix}
    0    &  0    &  0    &  0   & 0 \\
\theta_2    &  0    &  0    &  0   & \theta_2   \\
    0    &  \theta_3     &  0    &  0   &  0   \\
0      & 0      &     \theta_4      & 0   & 0 \\
0      &   0      &  0      & \theta_5& 0
\end{bmatrix}.\]

   \vspace*{1cm}

In the third case, with the the closed cycle of  length 5 (Figure \ref{fig_five_nodes} (c)), the dynamics of $\{A^{C^{out}}_{i,T}\}_{i=1}^5$ is governed by the following system

\begin{equation}
\begin{cases}
A^{C^{out}}_{1,T+1}= A^{0^{out}}_1 + (1-h) \theta_1  A^{C^{out}}_{5,T}\\
A^{C^{out}}_{2,T+1}= A^{0^{out}}_2+ (1-h)\theta_2 A^{C^{out}}_{1,T} \\
A^{C^{out}}_{3,T+1}= A^{0^{out}}_3+(1-h)\theta_3 A^{C^{out}}_{2,T} \\ 
A^{C^{out}}_{4,T+1}=  A^{0^{out}}_4  + (1-h) \theta_4 A^{C^{out}}_{3,T}\\ 
A^{C^{out}}_{5,T+1}=  A^{0^{out}}_5  + (1-h) \theta_5  A^{C^{out}}_{4,T}\\ 
\end{cases}
\end{equation}

We obtain 
\begin{equation}
\colvecfive[A^{C^{out}}_{1,T+1}]{A^{C^{out}}_{2,T+1}}{A^{C^{out}}_{3,T+1}} {A^{C^{out}}_{4,T+1}}{A^{C^{out}}_{5,T+1}}= \{ I+ [(1-h)M_c^1+ [(1-h) M_c]^2+...+ [(1-h)M_c]^T \}
\colvecfive[A^{0^{out}}_1]{ A^{0^{out}}_2}{A^{0^{out}}_3} {A^{0^{out}}_4}{A^{0^{out}}_5},
\label{eq_five_nodes_3}
\end{equation}
with 
 \[M_c=\begin{bmatrix}
    0    &  0    &  0    &  0   & \theta_1 \\
\theta_2    &  0    &  0    &  0   & 0 \\
    0    &  \theta_3     &  0    &  0   &  0   \\
0      & 0      &     \theta_4      & 0   & 0 \\
0      &   0      &  0      & \theta_5& 0
\end{bmatrix}.\]
It can be easily shown that given  $\{ \theta_i \}_{i}=\theta \ (\forall i)$ (i.e. under the condition of homogeneous hoarding), all panels $\alpha=a, b, c$  in Figure (\ref{fig_five_nodes})  create  the same amount of  $S_{\alpha, t}^{out}= \sum_{i=1}^{i=5} A^{C^{out}}_{i, t}$ for all $t\geq 1$.
This is because in all cases of $\alpha$   we have the same dynamics  $S_{\alpha, t+1}^{out}=S_{\alpha, 1}^{out}+  (1-h)\theta S_{\alpha, t}^{out}$  and they also all have the same initial aggregate amount of outgoing collateral, i.e.  $S_{\alpha, 1}^{out}=\theta  \sum_{i=1}^{i=5}A^{0}_{i}$. In addition, in  all panels  of Figure (\ref{fig_five_nodes}) we have that
$\lim_{t \to\infty}S_{\alpha, t}^{out}= \sum_{i=1}^{i=5} \frac {  A^{0^{out}}_i}{1-(1-h) \theta} =   \frac { \theta \sum_{i=1}^{i=5}A^{0}_{i}}{1-(1-h) \theta}$. Thus, at the fixed point solution to $A^{C^{out}}$, $m$ is equal to $\frac {1}{1-(1-h) \theta}$.

\bigskip

\textbf {Proof of proposition \ref{prop:dependence-cycles-length}}
  \medskip
In the case represented by  Figure \ref{fig_five_nodes_2} (a),  $k^{out}_5 =0 \rightarrow  \delta_5 =0$, therefore
\begin{equation}
\begin{cases}
A^{C^{out}}_{1,t=1}=A^{0^{out}}_1  = \theta_1 A^{0}_{1} \\ 
A^{C^{out}}_{2,t=1} =A^{0^{out}}_2 =  \theta_2 A^{0}_{2}   \\ 
A^{C^{out}}_{3,t=1} =A^{0^{out}}_3  =\theta_3 A^{0}_{3} \\ 
A^{C^{out}}_{2,t=1} =A^{0^{out}}_4 =  \theta_4 A^{0}_{4}  \\ 
A^{C^{out}}_{3,t=1} =A^{0^{out}}_5  =0\\ 
\end{cases}
\label{eq_five_nodes_initial_a}
\end{equation}
The dynamics of $\{A^{C^{out}}_{i,T}\}_{i=1}^5$ will follow
\begin{equation}
\begin{cases}
A^{C^{out}}_{1,T+1}= A^{0^{out}}_1 + (1-h) \theta_1  \frac {A^{C^{out}}_{3,T}}{2}\\
A^{C^{out}}_{2,T+1}= A^{0^{out}}_2+ (1-h)\theta_2 A^{C^{out}}_{1,T} \\
A^{C^{out}}_{3,T+1}= A^{0^{out}}_3+(1-h)\theta_3 A^{C^{out}}_{2,T} \\ 
A^{C^{out}}_{4,T+1}=  A^{0^{out}}_4  + (1-h) \theta_4 \frac {A^{C^{out}}_{3,T}}{2}\\ 
A^{C^{out}}_{5,T+1}= 0\\ 
\end{cases}
\label{eq_five_nodes_a1}
\end{equation}

It follows that
\begin{equation}
\colvecfive[A^{C^{out}}_{1,T+1}]{A^{C^{out}}_{2,T+1}}{A^{C^{out}}_{3,T+1}} {A^{C^{out}}_{4,T+1}}{A^{C^{out}}_{5,T+1}}=  [I+[(1-h)M_a]^1+ [(1-h)M_a]^2+...+[(1-h)M_a]^t
]\colvecfive[A^{0^{out}}_1]{ A^{0^{out}}_2}{A^{0^{out}}_3} {A^{0^{out}}_4}{0},
\label{eq_five_nodes_a2}
\end{equation}
with 
 \[M_a=\begin{bmatrix}
    0    &  0   &   \frac {\theta_1} {2}       &  0   & 0 \\
\theta_2   &  0 & 0    &  0   & 0 \\
    0    &  \theta_3   & 0   &  0   & 0  \\
0      & 0      &        \frac {\theta_4} {2}        & 0   & 0 \\
0      &   0      &  0      & 0 & 0
\end{bmatrix}.\]
In the case represented by  Figure \ref{fig_five_nodes_2} (b),  we still have that
\begin{equation}
\begin{cases}
A^{C^{out}}_{1,t=1}=A^{0^{out}}_1  = \theta_1 A^{0}_{1} \\ 
A^{C^{out}}_{2,t=1} =A^{0^{out}}_2 =  \theta_2 A^{0}_{2}   \\ 
A^{C^{out}}_{3,t=1} =A^{0^{out}}_3  =\theta_3 A^{0}_{3} \\ 
A^{C^{out}}_{2,t=1} =A^{0^{out}}_4 =  \theta_4 A^{0}_{4}  \\ 
A^{C^{out}}_{3,t=1} =A^{0^{out}}_5  =0\\ 
\end{cases}
\label{eq_five_nodes_initial_b}
\end{equation}
The dynamics of $\{A^{C^{out}}_{i,T}\}_{i=1}^5$ will follow
\begin{equation}
\begin{cases}
A^{C^{out}}_{1,T+1}= A^{0^{out}}_1 + (1-h) \theta_1  \frac {A^{C^{out}}_{4,T}}{2}\\
A^{C^{out}}_{2,T+1}= A^{0^{out}}_2+ (1-h)\theta_2 A^{C^{out}}_{1,T} \\
A^{C^{out}}_{3,T+1}= A^{0^{out}}_3+(1-h)\theta_3 A^{C^{out}}_{2,T} \\ 
A^{C^{out}}_{4,T+1}=  A^{0^{out}}_4  + (1-h) \theta_4 A^{C^{out}}_{3,T}\\ 
A^{C^{out}}_{5,T+1}=  0\\ 
\end{cases}
\label{eq_five_nodes_b1}
\end{equation}

We obtain 
\begin{equation}
\colvecfive[A^{C^{out}}_{1,T+1}]{A^{C^{out}}_{2,T+1}}{A^{C^{out}}_{3,T+1}} {A^{C^{out}}_{4,T+1}}{A^{C^{out}}_{5,T+1}}=  [I+ [(1-h)M_b]^1+ [(1-h)M_b]^2+...+[(1-h)M_b]^t
]\colvecfive[A^{0^{out}}_1]{ A^{0^{out}}_2}{A^{0^{out}}_3} {A^{0^{out}}_4}{0},
\label{eq_five_nodes_b2}
\end{equation}
with 
 \[M_b=\begin{bmatrix}
    0    &  0   &   0  &  \frac {\theta_1} {2}        & 0 \\
\theta_2   &  0 & 0    &  0   & 0 \\
    0    &  \theta_3   & 0   &  0   & 0  \\
0      & 0      &        \theta_4      & 0   & 0 \\
0      &   0      &  0      & 0  & 0
\end{bmatrix}.\]
Furthermore, in the case represented by  Figure \ref{fig_five_nodes_2} (c), as shown in (\ref{eq_five_nodes_initial}) and  (\ref{eq_five_nodes_3}),  we have 
\[
\begin{cases}
A^{C^{out}}_{1,t=1}=A^{0^{out}}_1  = \theta_1 A^{0}_{1} \\ 
A^{C^{out}}_{2,t=1} =A^{0^{out}}_2 =  \theta_2 A^{0}_{2}   \\ 
A^{C^{out}}_{3,t=1} =A^{0^{out}}_3  =\theta_3 A^{0}_{3} \\ 
A^{C^{out}}_{2,t=1} =A^{0^{out}}_4 =  \theta_4 A^{0}_{4}  \\ 
A^{C^{out}}_{3,t=1} =A^{0^{out}}_5  =\theta_5 A^{0}_{5}\\ 
\end{cases}
\]
Thus
\[
\colvecfive[A^{C^{out}}_{1,T+1}]{A^{C^{out}}_{2,T+1}}{A^{C^{out}}_{3,T+1}} {A^{C^{out}}_{4,T+1}}{A^{C^{out}}_{5,T+1}}= \{ I+ [(1-h)M_c]^1+ [(1-h) M_c]^2+...+ [(1-h)M_c]^T \}
\colvecfive[A^{0^{out}}_1]{ A^{0^{out}}_2}{A^{0^{out}}_3} {A^{0^{out}}_4}{A^{0^{out}}_5},
\label{eq_five_nodes_c}
\]
with 
 \[M_c=\begin{bmatrix}
    0    &  0    &  0    &  0   & \theta_1 \\
\theta_2    &  0    &  0    &  0   & 0 \\
    0    &  \theta_3     &  0    &  0   &  0   \\
0      & 0      &     \theta_4      & 0   & 0 \\
0      &   0      &  0      & \theta_5& 0
\end{bmatrix}.\]
We can see that the main difference between  Figure \ref{fig_five_nodes_2} (a) and  Figure \ref{fig_five_nodes_2} (b)   is mathematically expressed by the difference between  $M_a$ and  $M_b$: in the former case, a part of the initial outgoing collateral from the bank $4$ is flowing into the cycle, while in the later case all outgoing collateral from the bank $4$  will stuck in the box of  the bank $5$  and can not be re-used.  In addition, comparing these cases to the one represented by  Figure \ref{fig_five_nodes_2} (c), we can see that  in   Figure \ref{fig_five_nodes_2} (c) the initial outgoing collateral from each bank can be re-used infinitely.
Defining    $S_{a, t}^{out}$,  $S_{b, t}^{out}$, and   $S_{c, t}^{out}$ are respectively the total amount of outgoing collateral of all banks in Figures \ref{fig_five_nodes_2} (a), (b), and (c) after t times of using and re-using collateral, it can be  proved by induction that  $S_{a, t}^{out} < S_{b, t}^{out} < S_{c, t}^{out} \ (\forall t\geq 2)$.  This implies that in contrast to the example illustrated in Figure (\ref{fig_five_nodes}),  now  longer cycles  will generate more endogenous  collateral.

 \bigskip

\textbf {Proof of proposition \ref{prop:proposition_k_reg}.}

\medskip

To illustrate this proposition, without loss of generality we show an example of  closed cycle of length $N=5$ in Figure (\ref{fig_five_nodes_add_cycles}).  We can add  arbitrary links or cycles of length $k=3, 4$ to the initial graph. In all examples shown in Figure (\ref{fig_five_nodes_add_cycles}), we always have that  $k_i^{out} >0$ for every bank $i$. 
Proposition  (\ref{prop:proposition_k_reg}) is therefore just a
special case of Proposition
\ref{prop:proposition_positive_out_degree_basic}, of which we will provide the proof later.

\begin{figure}[H]
\centering
\captionsetup[subfloat]{farskip=0pt,captionskip=0pt}
\includegraphics[width=15cm]{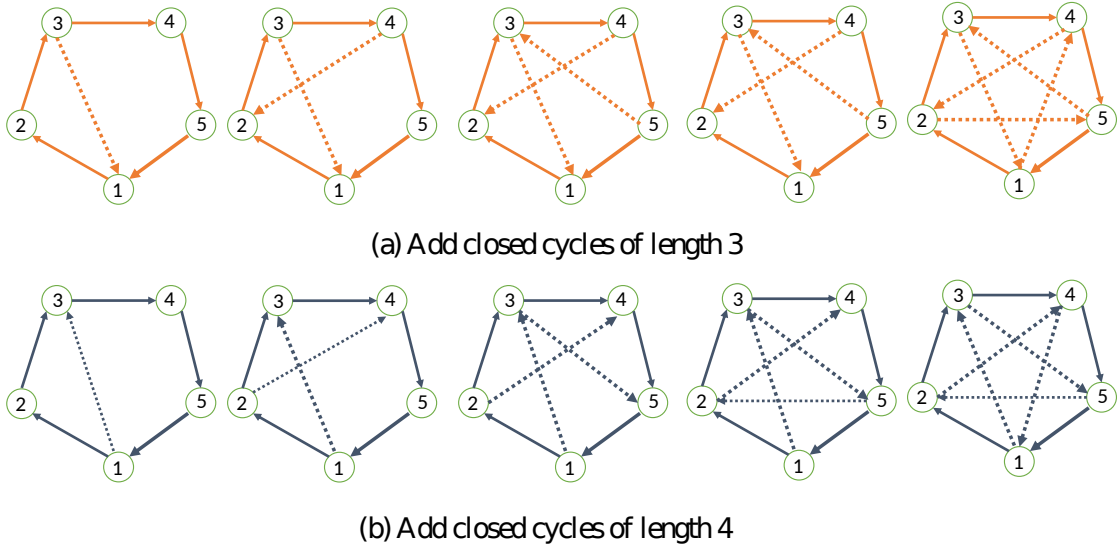}
\caption{Add different cycles (with length $k=3, ...N-1$) to an initial closed graph of size $N$. An example is illustrated with $N=5$ nodes. In the panel (a), we add cycles of length 3.  In the panel (b), we add cycles of length 4.}
\label{fig_five_nodes_add_cycles}
\end{figure}

  \bigskip
\textbf {Proof of proposition \ref{prop:proposition_positive_out_degree_basic}.}
  \medskip

First, since $k_i^{out} >0$, hence  $\delta_i= 1$  ($\forall
i$). Accordingly, we have $A^{C^{out}}_i= \theta
A^{C}_i$ and $ A^{0^{out}}_i= \theta  A^{0}_i$. Therefore,

\[
m= \frac {\sum_{i=1}^{i=N} A^{C^{out}}_i} {\sum_{i=1}^{i=N} A^{0^{out}}_i} = \frac { \theta  \sum_{i=1}^{i=N} A^{C}_i} {\theta \sum_{i=1}^{i=N} A^{0}_i}= \frac {\sum_{i=1}^{i=N} A^{C}_i} {\sum_{i=1}^{i=N} A^{0}_i}.
\]

In addition,  since each column\footnote{Now all columns are non-zero since  $k_i^{out} >0$ for all $i$.} of the matrix of shares  $\mathcal {S}= \{ s_{i\leftarrow j} \}_{NxN}$ is summing to 1 and thus $(1-h) \sum_{i \neq j} w_{i\leftarrow j}=(1-h)\theta$ ($\forall j$), at the fixed point solution to $A^C$, we always have   
\[ \sum_{i=1}^{i=N} A^C_i=(1-h)\theta \sum_{i=1}^{i=N} A^C_i +  \sum_{i=1}^{i=N} A^0_i.\] 
Equivalently,
\[\frac {\sum_{i=1}^{i=N} A^C_i}{\sum_{i=1}^{i=N} A^0_i}= \frac  {1}{1-(1-h)\theta}.\]
Hence,
\begin {equation}
 m=\frac {\sum_{i=1}^{i=N} A^{C^{out}}_i}  {\sum_{i=1}^{i=N} A^{0^{out}}_i}= \frac  {1}{1-(1-h)\theta}.
\label{m_limit_2}
\end {equation}

\bigskip

\textbf {Proof of proposition \ref{prop:proposition_density_rand}.}

   \vspace*{0.2cm}

Recall that under the basic model with homogeneous $\theta$ we have
\[
A^{C}_i=A^{0}_i+ (1-h) \theta  \sum_{j \in B_i} s_{i \leftarrow j} \delta_j A^{C}_j,
\]
where
\[
\begin{cases}
\delta_j= 1, \ \mbox {if} \  k_j^{out}>0\\ 
\delta_j= 0, \ \mbox {if} \  k_j^{out}=0 \\ 
\end{cases}
\]
Therefore, 
\begin {equation}
\E [A^{C}_i]=A^{0}_i+ (1-h) \theta  \E[\sum_{j \in B_i} s_{i \leftarrow j} \delta_j A^{C}_j],
\label{eq_ex_collateral_box_in}
\end {equation}
where the notation  $\E[X]$ stands for the expectation of $X$.  With the shares $s_{i \leftarrow j}$ defined as in equation (\ref{eq_ weights_size_2}) in the main text, we have
\begin {equation}
\E [A^{C}_i]= A^{0}_i+  (1-h) \theta   \sum_{j \in B_i} P( k_j^{in}>0)  \E[\frac {a_{i \leftarrow j}}{ k_j^{out}}  A^{C}_j],  
\label{eq_ex_collateral_box_in_1}
\end {equation}
with $P( k_j^{out}=0)$ is the probability that $ k_j^{out}=0$ and  $P( k_j^{out}>0)$ is the probability that $ k_j^{out}>0$. 

Under random graphs with a probability $p$ of a link between any two nodes (i.e.  $P(a_{i \leftarrow j}=1)  =p$), for every $j$ it is easy to show that $P( k^{out}_j=0)=P( k^{in}_j=0)= (1-p)^{(N-1)}$ and   $P( k^{out}_j>0)=P( k^{in}_j>0)= 1-(1-p)^{(N-1)}$.  In addition,  $\E[k_j^{out}]=\E[k_j^{in}]=p(N-1)$. That leads to
\begin {equation}
\E [A^{C}_i]= A^{0}_i+ [1-(1-p)^{(N-1)}] (1-h) \theta  \sum_{j \neq i} \E[  P(a_{i \leftarrow j}=1)   \frac {a_{i \leftarrow j}}{k_j^{out}}A^{C}_j].
\label{eq_ex_collateral_box_in_rg_1}
\end {equation}
From equation (\ref{eq_ex_collateral_box_in_rg_1}) we have\footnote {We use the approximation $\E[XY] \approx \E[X]\E[Y]$.} 
\begin {equation}
\E [A^{C}_i] \approx  A^{0}_i+ [1-(1-p)^{(N-1)}] (1-h) \theta   \frac {1}{N-1}\sum_{j \neq i} \E[ A^{C}_j]
\label{eq_ex_collateral_box_in_rg_2}
\end {equation}
and the following approximation\footnote{We use the approximation $\E[\frac{X}{Y}] \approx \frac{\E[X]}{\E[Y]}$.}  
\begin {equation}
\E [ \sum_{i=1}^{N} A^{C}_i] \approx \frac  {\sum_{i=1}^{N} A^{0}_i}{1-[1-(1-p)^{(N-1)}] (1-h) \theta}.
\label{eq_ex_collateral_box_in_rg_3}
\end {equation}
Thus,
\begin {equation}
 \lim_{p \to 1}  \frac {\E [ \sum_{i=1}^{N} A^{C}_i]} {\sum_{i=1}^{N} A^{0}_i} = \frac {1} {1-(1-h) \theta}.
\end {equation}

Note that
\begin {equation}
\begin {cases}
\E [A^{C^{out}}_i]= P( k_i^{out}>0) \theta \E [A^{C}_i] =(1-(1-p)^{(N-1)}) \theta \E [A^{C}_i]\\
\E [A^{0^{out}}_i]= P( k_i^{out}>0) \theta A^{0}_i=(1-(1-p)^{(N-1)}) \theta A^{0}_i\\
\end {cases}
\end {equation}
Consequently,
\begin {equation}
\E [ \sum_{i=1}^{N}A^{0^{out}}_i]= [1-(1-p)^{(N-1)}] \theta  \sum_{i=1}^{N} A^{0}_i\\
\label{eq_ex_collateral_box_out_rg_1}
\end {equation}
and
\begin {equation}
\E [ \sum_{i=1}^{N}A^{C^{out}}_i] \approx \frac {[1-(1-p)^{(N-1)}] \theta \sum_{i=1}^{N} A^{0}_i} {1-[1-(1-p)^{(N-1)}] (1-h) \theta}.\\
\label{eq_ex_collateral_box_out_rg_2}
\end {equation}
Hence,
\begin {equation}
 \E [m] \approx  \frac {\E [ \sum_{i=1}^{N} A^{C^{out}}_i]} {\E [  \sum_{i=1}^{N} A^{0^{out}}_i]}= \frac {1} {1-[1-(1-p)^{(N-1)}] (1-h) \theta}
\label{eq_ex_multiplier_collateral_box_out_RG}
\end {equation}
and
\begin {equation}
 \lim_{p \to 1}  \frac {\E [ \sum_{i=1}^{N} A^{C^{out}}_i]} {\E [  \sum_{i=1}^{N} A^{0^{out}}_i]}= \lim_{p \to 1}  \frac {1} {1-[1-(1-p)^{(N-1)}] (1-h) \theta} =\frac {1} {1-(1-h) \theta}.
\label{lim_RG}
\end {equation}

Moreover, 
we can easily show that the approximations for $\E [ \sum_{i=1}^{N} A^{C}_i]$, $\E [ \sum_{i=1}^{N}A^{0^{out}}_i]$, and $\E[m]$ are increasing functions of the network density $p$. In addition, given $p$ in $(0, 1]$, these measures are also increasing functions of the size of the network, $N$. 

  \vspace*{0.2cm}

We now proceed with proving the second part of the proposition. We consider a core-periphery graphs in which: 
(i)  the number of nodes in the core, $N_{core}$,  is fixed; (ii) each node in the periphery has only out-going links, and all point to nodes in the core; (iii) each node in the core are also randomly connected to  each other with the probability $p_{core}$, and there are no directed links from the core to the periphery nodes.

To begin, we first consider  the behavior of nodes in the periphery part  ($\text{Per}$). For every node $j$ in the periphery, we have
\[
\begin{cases}
 k_j^{out} \geq 1\\ 
\delta_j= 1 \\
\end{cases}
\]
and
\[
\begin{cases}
A^{C^{out}}_j= A^{0^{out}}_j= \theta A^{0}_j\\
A^{C}_j=A^{0}_j \\
\end{cases}
\]
Therefore,
\begin {equation}
\sum_{j \in \text{Per}} A^{C^{out}}_j =\theta  \sum_{j \in \text{Per}} A^{0}_j
\label{eq_ex_collateral_box_per_out}
\end {equation}
and
\begin {equation}
\sum_{j \in \text{Per}} A^{C}_j =\sum_{j \in \text{Per}} A^{0}_j. 
\label{eq_ex_collateral_box_per_in}
\end {equation}
Equations  (\ref{eq_ex_collateral_box_per_out}) and  (\ref{eq_ex_collateral_box_per_in}) imply that the aggregate amounts of in-flowing and out-going collateral of all nodes in the periphery remain constant during rehypothecation process. This observation is intuitive since we assume that periphery banks are purely borrowers and therefore they do not receive collateral  from other banks. 

Moving on to the core part,  for every bank $i$ in the core part ($\text{Core}$) we have
\begin {equation}
A^{C}_i=A^{0}_i + (1-h) \theta \sum_{ j \in \text{Per}} s_{i \leftarrow j} \delta_j A^{C}_j + (1-h) \theta \sum_{ j \in \text{Core}} s_{i \leftarrow j} \delta_j A^{C}_j.
\label{eq_ex_collateral_box_in_core_1}
\end {equation}
Taking the expectation from both sides, we have
\begin {equation}
\E[A^{C}_i]=A^{0}_i + (1-h) \theta \sum_{j \in \text{Per}}  s_{i \leftarrow j}  A^{0}_j + (1-h) \theta  \E[\sum_{ j \in \text{Core}} s_{i \leftarrow j} \delta_j A^{C}_j].
\label{eq_ex_collateral_box_in_core_2}
\end {equation}
Defining
\begin {equation}
\tilde {A}^{0}_i= A^{0}_i + (1-h) \theta \sum_{ j \in \text{Per}}  s_{i \leftarrow j}  A^{0}_j,
\label{eq_new_initial_collateral_core}
\end {equation}
then we have 
\begin {equation}
\E[A^{C}_i]=\tilde {A}^{0}_i + (1-h) \theta  \E[\sum_{j \in \text{Core}} s_{i \leftarrow j} \delta_j A^{C}_j].
\label{eq_ex_collateral_box_in_core_3}
\end {equation}
Equations (\ref{eq_new_initial_collateral_core}) and (\ref{eq_ex_collateral_box_in_core_1}) respectively imply two important characteristics of the rehypothecation of collateral under the considered core-periphery structure, i.e. the concentration into the core part and the reuse of collateral among banks in the core.

In addition,  since  each non-zero column of the matrix of shares  $\mathcal {S}= \{ s_{i\leftarrow j} \}_{N\text{x}N}$ is summing to 1, it  is easy to show that
\begin {equation}
\sum_{ i \in \text{Core}} \tilde {A}^{0}_i = \sum_{ i \in \text{Core}} A^{0}_i + (1-h) \theta \sum_{ j \in \text{Per}} A^{0}_j.
\label{eq_initial_box_in_core}
\end {equation}

Since nodes in the core are assumed to be randomly connected with the density $p_{core}$, using the results obtained from random graphs we have
\begin {equation}
\E [ \sum_{i \in \text{Core}} A^{C}_i]=\frac  {\sum_{i \in \text{Core}}\tilde {A}^{0}_i }{1-[1-(1-p)^{(N_{core}-1)}] (1-h) \theta}.
\label{eq_ex_collateral_box_in_core_4}
\end {equation}

Equations (\ref{eq_ex_collateral_box_per_in}) and (\ref{eq_ex_collateral_box_in_core_4}) lead to
\begin {equation}
\E [ \sum_{i=1}^{N} A^{C}_i]= \E [ \sum_{i \in \text{Per}} A^{C}_i] + \E [ \sum_{i \in \text{Core}} A^{C}_i]= \sum_{i \in \text{Per}} A^{0}_j + \frac  {\sum_{i \in \text{Core}}\tilde {A}^{0}_i }{1-[1-(1-p_{core})^{(N_{core}-1)}] (1-h) \theta}.
\label{eq_ex_collateral_box_in_CP}
\end {equation}

Note that, for every node $i$ in the core part
\begin {equation}
\begin {cases}
\E [A^{C^{out}}_i]= P( k_i^{out}>0) \theta \E [A^{C}_i] =[1-(1-p_{core})^{(N_{core}-1)}] \theta \E [A^{C}_i]\\
\E [A^{0^{out}}_i]= P( k_i^{out}>0) \theta A^{0}_i=[1-(1-p_{core})^{(N_{core}-1)}] \theta A^{0}_i\\
\end {cases}
\end {equation}

Therefore
\begin {equation}
\E [ \sum_{i \in \text{Core}} A^{C^{out}}_i]=\frac  {[1-(1-p_{core})^{(N_{core}-1)}]\theta \sum_{i \in \text{Core}}\tilde {A}^{0}_i }{1-[1-(1-p_{core})^{(N_{core}-1)}] (1-h) \theta}
\label{eq_ex_collateral_box_out_core_1}
\end {equation}
and
\begin {equation}
\E [ \sum_{i \in \text{Core}}A^{0^{out}}_i]=[1-(1-p_{core})^{(N_{core}-1)}] \theta \sum_{i \in \text{Core}}A^{0}_i.
\label{eq_ex_collateral_box_out_core_2}
\end {equation}

From  (\ref{eq_ex_collateral_box_per_out}) and (\ref{eq_ex_collateral_box_out_core_1}) we have
\begin {equation}
\E [ \sum_{i=1}^{N} A^{C^{out}}_i]= \E [ \sum_{j \in \text{Per}} A^{C^{out}}_j] + \E [ \sum_{i \in \text{Core}} A^{C^{out}}_i]=\theta  \sum_{j \in \text{Per}} A^{0}_j + \frac  {[1-(1-p_{core})^{(N_{core}-1)}]\theta \sum_{i \in \text{Core}}\tilde {A}^{0}_i }{1-[1-(1-p_{core})^{(N_{core}-1)}] (1-h) \theta},
\label{eq_ex_collateral_box_out_CP_1}
\end {equation}
which can be simplified as 
\begin {equation}
\E [ \sum_{i=1}^{N} A^{C^{out}}_i]= \frac  {\theta  \sum_{j \in \text{Per}} A^{0}_j + [1-(1-p_{core})^{(N_{core}-1)}]\theta \sum_{i \in \text{Core}}A^{0}_i }{1-[1-(1-p_{core})^{(N_{core}-1)}] (1-h) \theta}
\label{eq_ex_collateral_box_out_CP}
\end {equation}
by substituting $\tilde {A}^{0}_i$  in equation (\ref {eq_initial_box_in_core}) into equation (\ref{eq_ex_collateral_box_out_CP_1}).

Moreover, the expectation of the aggregate amount of initial out-going collateral is
\begin {equation}
\E [ \sum_{i=1}^{N} A^{0^{out}}_i]= \E [\sum_{j \in \text{Per}} A^{0^{out}}_j]+ \E [\sum_{i \in \text{Core}} A^{0^{out}}_i]=\theta  \sum_{j \in \text{Per}} A^{0}_j + [1-(1-p_{core})^{(N_{core}-1)}]\theta \sum_{i \in \text{Core}}A^{0}_i.
\label{eq_ex_collateral_box_out_CP_2}
\end {equation}
As results,
\begin {equation}
 \E [m] \approx  \frac {\E [ \sum_{i=1}^{N} A^{C^{out}}_i]} {\E [  \sum_{i=1}^{N} A^{0^{out}}_i]}= \frac {1} {1-[1-(1-p_{core})^{(N_{core}-1)}] (1-h) \theta}
\label{eq_ex_multiplier_collateral_box_out_CP}
\end {equation}
and
\begin {equation}
 \lim_{p_{core} \to 1}  \frac {\E [ \sum_{i=1}^{N} A^{C^{out}}_i]} {\E [  \sum_{i=1}^{N} A^{0^{out}}_i]}= \frac {1}{1- (1-h) \theta}.\label{eq_lim_CP}
\end {equation}

Again, it can be shown that $\E [ \sum_{i=1}^{N} A^{C}_i]$, $\E [ \sum_{i=1}^{N}A^{0^{out}}_i]$ and the approximation for $\E[m]$ are increasing functions of the density of the core $p_{core}$. In addition, these three measures are also increasing functions of the size of the core part, $N_{core}$, given $p_{core}$  in $(0, 1]$. 
Comparing the multiplier estimated for core-periphery graphs, as in equation (\ref{eq_ex_multiplier_collateral_box_out_CP}) and with the one estimated for random graphs as in equation (\ref{eq_ex_multiplier_collateral_box_out_RG}), we notice that the former is always larger than the latter, as long as $p_{core}> 1-(1-p)^{\frac{N-1}{N_{core}-1}}= p_{th}$, thus verifying the third part of the proposition. Finally, by inspecting equations (\ref{lim_RG}) and (\ref{eq_lim_CP}) we verify the fourth and final part of the proposition.

  \vspace*{0.2cm}
\textbf {Proof of proposition \ref{proposition_solution_A_C_VaR}.}
  \vspace*{0.2cm}

We will now provide detailed derivations for the equilibrium existence to non-hoarded parameters determined under the Value-at-Risk strategy.
To begin,  let us start with the following definitions:
\begin{definition}
 For each column vector $X=[X_1, X_2,...,  X_{N-1}, X_N]^T\in\mathbb{R}_{+}^N $,  $X \geq \frac {1}{(1-h)}C^0$  (where  $\mathcal {C}^0=[c^0_1, c^0_2,..., c^0_{N-1}, c^0_N]^T$),  the elements of the matrix $ \mathcal {\tilde {W}}^{\text{VaR}}(X)$ (size $N\text{x}N$) is defined as
\begin {equation}
\begin {cases}
 \mathcal {\tilde {W}}_{i \leftarrow j}^{\text{VaR}}=\frac {(1-h)\text {X}_j- c^0_j}{\text {X}_j \text {k}_j^{out} } , \ \forall i \in L_j \neq  \varnothing\\
 \mathcal { \tilde {W}}_{i \leftarrow  j}^{\text{VaR}}= 0,\  \mbox {elsewhere}\\
\end {cases}
\end {equation}
\end{definition}

We now show that the following system of equations
\begin{equation}
X=A^0+  \mathcal {\tilde {W}}^{\text {VaR}}(X) X
\label{eq_collateral_box_VaR_appendix}
\end{equation}
 has a single unique solution if $0< \text {h} < 1$ and $\mathcal {\tilde {A}}^{-1}\tilde {b} \geq \frac { \mathcal{C}^0}{1-h}$, where  $\tilde {b}$ is a column vector size $N \text{x} 1$,  with
\begin {equation}
\tilde {b}= A^0- \mathcal {S}  \mathcal{C}^0
\end{equation}
and 
\begin {equation}
 \mathcal {\tilde {A}}  =\mathcal {I}-(1-h) \mathcal {S},
\end{equation}
with  the matrix of shares, $\mathcal {S}$, is defined as in equation (\ref{eq_ weights_size_2})  in
 the main text.

Notice that
\[X=A^0+  \mathcal {\tilde {W}} ^{\text {VaR}}(X) X\]
$ \Leftrightarrow $
\begin {equation}
X_i= A^0_i+ \sum_{j \neq i} \frac {a_{i \leftarrow j}}{k_j^{out}} [(1-h) X_j-c^0_j], \forall i=1,2,...N
\end {equation}
$ \Leftrightarrow $
\begin {equation}
 A^0_i - \sum_{j \neq i} \frac {a_{i\leftarrow j}}{k_j^{out}} c^0_j =X_i- \sum_{j \neq i} \frac {(1-h) a_{i\leftarrow j}}{k_j^{out}} X_j,\forall i=1,2,...N.
\label{equation:fix_point_VaR_appendix1}
\end {equation}
Equivalently,
\begin {equation}
 \tilde {b}= \mathcal {\tilde {A}} X
\label{equation:fix_point_VaR_appendix2}.
\end {equation}
If $\mathcal {\tilde {A}}$  is an invertible matrix,  system (\ref {equation:fix_point_VaR_appendix2}) has a single unique solution
\begin {equation}
X= \mathcal {\tilde {A}} ^{-1} \tilde {b}.
\end {equation}
We now show the invertibility of $\mathcal {\tilde {A}}$ by reductio ad absurdum. Suppose that $\mathcal {\tilde {A}}$ is not invertible, then $det( \mathcal {\tilde {A}} )=0$.
Note that  $ \mathcal {\tilde {A}} = I-(1-h) \mathcal {S}$ therefore $det( \mathcal {\tilde {A}} )=0$ if and only if $\frac {1}{1-h} (>1)$ is an eigenvalue of $\mathcal {S}$ . However, since  each non-zero column of   $\mathcal {S}$ is summing to 1, according to  Perron-Frobenius theorem, the largest eigenvalue of  $\mathcal {S}$ can not be  larger than  1.

}

\end{document}